\documentclass[11pt]{amsart}
\title{Graded colour Lie superalgebras for solving L\'evy-Leblond equations}
\author[M.~Ryan]{Mitchell Ryan}

\usepackage[margin=1in]{geometry}
\usepackage[foot]{amsaddr}
\usepackage{amsmath,amssymb,amsthm}
\usepackage{mathtools}
\usepackage[mathscr]{euscript}
\usepackage{physics}
\usepackage{enumitem}
\usepackage[colorlinks=true, pdfstartview=FitV, linkcolor=blue, citecolor=blue, urlcolor=blue]{hyperref}
\usepackage{cite}
\usepackage{cleveref}
\usepackage{stmaryrd}
\usepackage{xcolor}

\DeclareDocumentCommand\cbrak{ l m m }{\braces#1{\llbracket}{\rrbracket}{#2,#3}} 
\newcommand{\adjt}[1]{#1^{\dagger}}				
\newcommand{\chiralg}{\widetilde{\gamma}^{\mathrm{chir}}}	
\newcommand{\conj}[1]{#1^*}					
\newcommand{\CC}{\mathbb{C}}					
\newcommand{\Cl}[2]{C\ell_{#1,#2}}				
\newcommand{\g}{\mathfrak{g}}					
\newcommand{\Hll}{H_{\textup{LL}}}				
\newcommand{\Hsch}{H_{\textup{SL}}}				
\newcommand{\Id}{I}						
\renewcommand{\laplacian}{\Delta}				
\newcommand{\Lp}[1]{L^{#1}}					
\newcommand{\parity}{\mathscr{P}}				
\newcommand{\osp}{\mathfrak{osp}}				
\newcommand{\ocmpl}[1]{#1^{\perp}}				
\newcommand{\RR}{\mathbb{R}}					
\DeclareMathOperator{\spn}{span}				
\newcommand{\vac}{\ket{0}}					
\newcommand{\vacsch}{\ket{0_{\textup{SL}}}}			
\newcommand{\vacschp}{\ket{0'_{\textup{SL}}}}			
\let\vec\mathbf							
\newcommand{\ZZ}{\mathbb{Z}}					
\newcommand{\Ztwo}[1][]{\ZZ_2^{#1}}				
\newcommand{\Ztzt}{\Ztwo\times\Ztwo}				

\newcommand{\dfn}[1]{{\color{red!70!black}\itshape #1}}

\newtheorem{thm}{Theorem}[section]

\newtheorem{prop}[thm]{Proposition}

\theoremstyle{definition}

\newtheorem{ex}[thm]{Example}

\theoremstyle{remark}
\newtheorem{rmk}[thm]{Remark}

\begin{document}
\address{School of Mathematics and Physics, University of Queensland, St.\ Lucia, QLD 4072, Australia\\
ORCID: 0009-0006-2038-4410}
\email{\href{mailto:mitchell.ryan@uq.edu.au}{mitchell.ryan@uq.edu.au}}


\begin{abstract}
	The L\'evy-Leblond equation with free potential admits a symmetry algebra that is a \( \mathbb{Z}_2\times\mathbb{Z}_2 \)-graded colour Lie superalgebra
	(see Aizawa--Kuznetsova--Tanaka--Toppan, 2016).
	We extend this result in two directions
	by considering a time-independent version of the L\'evy-Leblond equation.
	First,
	we construct a \( \mathbb{Z}_2^3 \)-graded colour Lie superalgebra containing operators that leave the eigenspaces invariant
	and demonstrate the utility of this algebra in constructing general solutions for the free equation.
	Second,
	we find that the ladder operators for the harmonic oscillator generate a \( \mathbb{Z}_2\times\mathbb{Z}_2 \)-graded colour Lie superalgebra
	and we use the operators from this algebra to compute the spectrum.
	These results illustrate two points: the L\'evy-Leblond equation admits colour Lie superalgebras with gradings higher than \( \mathbb{Z}_2\times\mathbb{Z}_2 \) and colour Lie superalgebras appear for potentials besides the free potential.
\end{abstract}
\maketitle

\section{Introduction}
The L\'evy-Leblond equation is the non-relativistic limit of the Dirac equation~\cite{LevyLeblond1967}.
Dirac introduced his equation 
to describe a relativistic spin-\( 1/2 \) particle
and was the first quantum-mechanical equation to account for both spin and special relativity~\cite{Dirac1928}.
To obtain his equation, Dirac took a `square root' of the Klein--Gordon equation (which describes a relativistic spin-\( 0 \) particle) in order to replace the second-order time derivative with a first-order derivative~\cite{Dirac1928}.
This square root required the introduction of gamma matrices (which generate Clifford algebras%
) and naturally led to the introduction of spin into the equation.

In a similar way,
L\'evy-Leblond~\cite{LevyLeblond1967} obtained his equation as a `square root' of the non-relativistic Schr\"odinger equation.
Similar to the Dirac equation, this square root required the introduction of gamma matrices and spin.
L\'evy-Leblond showed that 
his equation is invariant under the Galilei group,
predicts the correct value for the magnetic moment of a spin-\( 1/2 \) particle
and can be obtained as the non-relativistic limit of the Dirac equation~\cite{LevyLeblond1967}.

Remarkably, it has since been shown that the L\'evy-Leblond equation has symmetry algebras which are \( \Ztzt \)-graded colour Lie superalgebras~\cite{AKTT2016,AKTT2017} .
The authors of~\cite{AKTT2016,AKTT2017} were motivated by the fact that
ungraded Schr\"odinger operators can give rise to \( \Ztwo \)-graded Lie superalgebras (see~\cite{Toppan2015});
for example,
the ladder operators for the quantum harmonic oscillator generate a \( \Ztwo \)-graded Lie superalgebra \( \osp(1|2) \).
With this motivation~\cite{AKTT2016}, 
it was investigated whether replacing the ungraded Schr\"odinger operator
with the \( \Ztwo \)-graded L\'evy-Leblond operator
would yield symmetry with a second \( \Ztwo \)-gradation, making the symmetry algebras \( \Ztzt \)-graded.
Such \( \Ztzt \)-graded symmetry algebras did appear for the free equation~\cite{AKTT2016,AKTT2017}
in the form of a colour Lie superalgebra.
However, this \( \Ztzt \)-graded symmetry disappeared when considering the harmonic potential~\cite{AKTT2016}.

Colour Lie (super)algebras were first introduced in~\cite{RW1978a,RW1978b} (however, see also~\cite{Ree1960})
and generalise Lie superalgebras to grading by any abelian group.
Since their application to the L\'evy-Leblond equation~\cite{AKTT2016,AKTT2017}, there has been an increase in research activity surrounding applications of colour Lie (super)algebras.
For instance, colour Lie (super)algebras have been applied to
generalised quantum mechanics~\cite{BD2020a,AKT2020,AKT2021,AAD2020a,AAD2020b,DA2021,AAD2021,Bruce2021,Quesne2021,AIKT2023} 
and parastatistics~\cite{YJ2001,JYL2001,KHA2011a,KHA2011b,Tolstoy2014b,SVdJ2018,Toppan2021a,Toppan2021b,Zhang2023,SVdJ2024}.
There also exists a largely independent line of research focussing on graded (non-colour) Lie algebras%
---see~\cite{BP2009,Ryan2024} for applications of results for such graded (non-colour) Lie algebras to the study of colour Lie (super)algebras.

Despite the recent activity surrounding colour Lie (super)algebras,
the \( \Ztzt \)-graded nature of the L\'evy-Leblond equation has received little attention since the first two papers~\cite{AKTT2016,AKTT2017} (see \cite{BCDM2020,Faustino2023} for recent work).
This paper addresses two questions raised by the work in~\cite{AKTT2016,AKTT2017}:
\begin{enumerate}
	\item\label{question:differentgradings} Do colour Lie superalgebras with gradings different from \( \Ztzt \) appear?
	\item\label{question:nonfreepotential} Do colour Lie superalgebras appear for potentials besides the free potential?
\end{enumerate}
We demonstrate an affirmative answer to both questions.
However, unlike~\cite{AKTT2016,AKTT2017},
we will not be focussed on symmetry algebras.
Instead, we search for algebras that aid in solving the time-independent equation.

The time independent equation is of the form \( \Hll\ket{\psi} = \gamma_+E\ket{\psi} \).
Here, \( \gamma_+ \) is a gamma matrix, \( E\in\CC \), \( \ket{\psi} \) is a state in the Hilbert space and \( \Hll \) is a particular matrix differential operator which we treat as if it were a Hamiltonian.
Of particular note is the fact that \( \gamma_+ \) is not invertible, and we are thus forced to leave \( \gamma_+ \) on the right-hand side of the equation.
Regardless, we proceed with the analysis as normal, and call \( E \) a `\( \gamma_+ \)-eigenvalue'.

For question~\ref{question:differentgradings}, we examine a time-independent version of the \( 1+1 \)-dimensional L\'evy-Leblond equation with free potential.
After solving a differential equation, we find five linearly independent operators which leave the \( \gamma_+ \)-eigenspaces invariant: the identity, a gamma matrix multiplied by the parity operator, the Schr\"odinger Hamiltonian, and two different `square roots' of the Schr\"odinger Hamiltonian.
These five operators close to form a \( \Ztwo[3] \)-graded colour Lie superalgebra \( \mathfrak{D} \).
We show that this \( \Ztwo[3] \)-graded colour Lie superalgebra is fundamental to the L\'evy-Leblond equation:
solutions to the L\'evy-Leblond equation can be expressed in terms of simultaneous \( \gamma_+ \)-eigenstates of the two `square roots' of the Schr\"odinger Hamiltonian.
As an illustration, we use this fact to solve the L\'evy-Leblond equation in this simple case.

To answer question~\ref{question:nonfreepotential},
we take a slightly different approach from~\cite{AKTT2016} and introduce the harmonic potential in a different form.
We then set up a \( \gamma_+ \)-eigenvalue equation and search for ladder operators.
We find that the ladder operators with the L\'evy-Leblond Hamiltonian generate a \( \Ztzt \)-graded colour Lie superalgebra.
This colour Lie superalgebra is the direct analogue of the \( \osp(1|2) \) superalgebra for the Schr\"odinger harmonic oscillator but, in the case of the L\'evy-Leblond equation, an extra \( \Ztwo \) gradation is required.
We then use the \( \Ztzt \)-graded colour Lie superalgebra to compute the spectrum.

This paper is organised as follows:
in~\Cref{sec:LLequation} we introduce the L\'evy-Leblond equation 
and set up a time-independent version of the equation.
In~\Cref{sec:free},
we construct the \( \Ztwo[3] \)-graded colour Lie superalgebra of \( \gamma_+ \)-eigenspace-preserving operators of the \( (1+1) \)-dimensional free equation.
In~\Cref{sec:harmonic}, we construct the \( \Ztzt \)-graded colour Lie superalgebra
generated by the ladder operators of the \( (1+1) \)-dimensional L\'evy-Leblond harmonic oscillator.
In \Cref{sec:NonrelativisticLimit}, we show how the gamma matrix relations used in this paper can be obtained when considering the L\'evy-Leblond equation as the non-relativistic limit of the Dirac equation.

\section{The L\'evy-Leblond equation and colour Lie superalgebras}\label{sec:LLequation}
\subsection{Free equation}
For the form of the free L\'evy-Leblond equation, we will follow~\cite{AKTT2016}.
A  L\'evy-Leblond equation is a first-order differential equation that is a square root of the heat or Schr\"odinger equation in \( (1+d) \)-dimensions \cite{AKTT2016}. 
The goal of the L\'evy-Leblond equation is to reduce a second-order differential equation to a first-order equation.

Let \( t \) be the time coordinate, let \( x_j \) be the \( j \)-th spacial coordinate and take \( \lambda\in\CC \).
Throughout this paper, we will use the notation \( \vec{x} = (x_1,\ldots,x_d) \) and \( \partial_t = \pdv{t}\), \( \partial_j = \pdv{x_j} \).
The \dfn{L\'evy-Leblond operator} (with free potential) is
\begin{equation}\label{eq:generalLevyLeblond}
	\Omega = \gamma_+ \partial_t +  \gamma_-\lambda + \gamma^j \partial_j
\end{equation}
where Einstein summation convention is used,
and the coefficients \( \gamma_+,\gamma_-, \gamma^j \:\) are gamma matrices satisfying
\begin{equation}\label{eq:LevyLeblondAR}
	\begin{aligned}
		\acomm{\gamma_{\pm}}{\gamma_{\pm}} &= 0, & \acomm{\gamma_+}{\gamma_-} &= \Id, \\
		\acomm{\gamma_{\pm}}{\gamma^j} &= 0, & \acomm{\gamma^j}{\gamma^k} &= 2\delta_{jk},
	\end{aligned}
\end{equation}
(with \( \delta_{jk} \) the Kronecker delta).

Using \eqref{eq:LevyLeblondAR}, the square of the L\'evy-Leblond operator is
\[
	\Omega^2 = \lambda\partial_t + \laplacian
\]
where \( \laplacian = \sum_j \partial_j^2 \) is the Laplacian.
If \( \lambda \) is a negative real number,
then the partial differential equation induced by \( \Omega^2 \)
(i.e.\ \( \Omega^2 \Psi(t,\vec{x}) = 0 \))
becomes the heat equation,
and if \( \lambda = i\beta \) for \( \beta\in\RR_{>0} \) (in particular, \( \beta = 2m/\hbar \) where \( m \) is mass and \( \hbar \) is the reduced Planck constant) then the partial differential equation induced by \( \Omega^2 \) becomes the Schr\"odinger equation.

The gamma matrices in the L\'evy-Leblond equation belong to some Clifford algebra.
The most general Clifford algebra required is \( \Cl{2}{d}(\RR)\otimes\CC \).
However, as in~\cite[Section~4]{AKTT2016}, we require that the gamma matrices for the \( (1+1) \)-dimensional equation also satisfy the following relations
\begin{equation}\label{eq:gammaMatricesFreeEquation}
	(\gamma_{\pm})^2 = 0, \qquad (\gamma^1)^2 = \Id, \qquad \gamma_{\pm}\gamma_{\mp} = \frac{1}{2}(\Id \pm \gamma^1), \qquad \gamma^1 \gamma_{\pm} = \pm \gamma_{\pm} = - \gamma_{\pm}\gamma^1.
\end{equation}
In this case, the gamma matrices can be chosen as elements of the Clifford algebra \( \Cl{1}{1}(\RR)\otimes\CC \).
More generally, we argue in \Cref{sec:NonrelativisticLimit} that \( \Cl{1}{d}(\RR)\otimes\CC \) is the most suitable Clifford algebra when \( d \) is odd.

\begin{rmk}
	We could work entirely with real matrices, and realise the imaginary unit as a real matrix \( J \) which commutes with all the gamma matrices and squares to the negative identity matrix.
	This would allow us to realise the gamma matrices as a representation of a \emph{real} Clifford algebra, and is the approach taken in~\cite{AKTT2016}.
	However, for simplicity, we will allow our gamma matrices to be over the complex numbers.
\end{rmk}

To actually solve the L\'evy-Leblond equation, we would need to choose a specific matrix representation of the corresponding Clifford algebra (though many of the properties that we will study are representation independent).
This makes the L\'evy-Leblond equation a matrix differential equation.
In particular, if we realise relations~\eqref{eq:LevyLeblondAR} using \( n \times n \)-dimensional matrices,
the corresponding Hilbert space of the quantum system is \( \Lp{2}(\RR)\otimes \CC^n \),
whose elements we interpret as \( n \)-component vectors of functions.
This interpretation provides a natural action of the L\'evy-Leblond operator on the Hilbert space.
For example, in \( (1+1) \)-dimensions, such a representation (as used in~\cite{AKTT2016}) is
\begin{equation}\label{eq:gammamatrixrep}
	\gamma^1 = 
	\begin{pmatrix}
		1 & 0\\
		0 & -1
	\end{pmatrix},
	\qquad 
	\gamma_+ =
	\begin{pmatrix}
		0 & 1\\
		0 & 0
	\end{pmatrix},
	\qquad
	\gamma_- =
	\begin{pmatrix}
		0 & 0\\
		1 & 0
	\end{pmatrix}
\end{equation}
which acts on the Hilbert space
\[
	\Lp{2}(\RR)\otimes \CC^2
	= \left\{
		\begin{pmatrix}
			f\\
			g
		\end{pmatrix}
		\mid
		f,g\in\Lp{2}(\RR)
	\right\}
\]
(note that this representation satisfies~\eqref{eq:gammaMatricesFreeEquation}).
If \( \Psi(t,\vec{x}) \) is a solution to the free L\'evy-Leblond equation (i.e.\ \( \Omega \Psi(t,\vec{x}) = 0 \)),
then we also have \( \Omega^2 \Psi(t,\vec{x}) = 0 \), so each component of \( \Psi(t,\vec{x}) \) will be a solution to the Schr\"odinger equation.

\subsection{L\'evy-Leblond operator with potential}\label{sec:potential}
Up to this point, we have only been considering the free equation.
There are (at least) two ways to add a potential function to the L\'evy-Leblond operator.

From a supersymmetry perspective, a natural choice for the L\'evy-Leblond operator with potential is the supercharge
\begin{equation}\label{eq:LevyLeblondSupercharge}
	\Omega = \gamma_+\partial_t + \gamma_- \lambda + \gamma^j\partial_j + \gamma_f f(\vec{x})
\end{equation}
for some additional gamma matrix \( \gamma_f \)
that satisfies additional anticommutation relations (omitted; see~\cite{AKTT2016}). 
These anticommutation relations for \( \gamma_f \) ensure that the square of the L\'evy-Leblond operator is
\[
	\Omega^2 = \lambda\partial_t + \laplacian + (f(\vec{x}))^2 + \gamma^j\gamma_f\pdv{f}{x_j}.
\]
In this case, the potential function is the square function \( (f(\vec{x}))^2 \).
This is the choice of L\'evy-Leblond operator made in~\cite{AKTT2016}.

From a gauge transformation perspective,
the L\'evy-Leblond operator can be obtained by replacing \( i\hbar\partial_t \) with \( i\hbar \partial_t-V(\vec{x}) \)
(where \( V(\vec{x}) \) is the potential function).
Working in natural units \( \hbar = 1 \) and multiplying through by \( i \),
the L\'evy-Leblond operator then becomes
\begin{equation}\label{eq:LevyLeblondPotential}
	\Omega = \gamma_+(i\partial_t - V(\vec{x})) - \gamma_- \beta + \gamma^ji\partial_j
\end{equation}
with square
\[
	\Omega^2 = -i\beta\partial_t - \laplacian + \beta V(\vec{x}) - \gamma^j\gamma_+ i \pdv{V}{x_j}.
\]
This was the choice of operator originally made by L\'evy-Leblond~\cite{LevyLeblond1967} 
(in combination with a magnetic vector potential which we have omitted).

When  compared to the Schr\"odinger operator, the squares of both~\eqref{eq:LevyLeblondSupercharge} and~\eqref{eq:LevyLeblondPotential}
contain an additional term which depends on the derivative of a potential.
Thus, unlike the free equation, the components of a solution to the L\'evy-Leblond equation with non-vanishing potential are not necessarily a solution to the Schr\"odinger equation.

In this paper, we will use the operator~\eqref{eq:LevyLeblondPotential} with harmonic potential,
since the colour Lie superalgebra symmetries of~\eqref{eq:LevyLeblondSupercharge} with harmonic potential have already been studied in \cite{AKTT2016}.
The operator~\eqref{eq:LevyLeblondPotential} has the advantages that it does not require the introduction of an additional gamma matrix and that we can work directly with the potential function instead of a `square-root' function \( f \) with \( (f(\vec{x}))^2 = V(\vec{x}) \).
One disadvantage of~\eqref{eq:LevyLeblondPotential} is that, with a quadratic potential \( V(x_1)\propto {x_1}^2 \), the square of~\eqref{eq:LevyLeblondPotential} contains an extra \( \gamma^1\gamma_+x_1 \) term, whereas the square of~\eqref{eq:LevyLeblondSupercharge} only has a extra \( \gamma^1\gamma_f \) term that is independent of \( x_1 \).

\subsection{Time-independent L\'evy-Leblond equation}
The differential equation induced by \eqref{eq:LevyLeblondPotential} is 
\begin{equation}\label{eq:TimeDependentLevyLeblond}
	(\gamma_- \beta - \gamma^ji\partial_j + \gamma_+V(\vec{x}))\Psi(t,\vec{x}) = \gamma_+i\partial_t \Psi(t,\vec{x}).
\end{equation}
We can choose to interpret  \( \gamma_-\beta - \gamma^j i \partial_j + \gamma_+V(\vec{x}) \eqqcolon \Hll \) as if it were a Hamiltonian and set up a time-independent version of this equation:
\begin{equation}\label{eq:TimeIndependentLevyLeblond}
	(\gamma_- \beta - \gamma^j i \partial_j + \gamma_+ V(\vec{x})) \psi(\vec{x}) = \gamma_+ E \psi(\vec{x})
\end{equation}
for some \( E\in\CC \).
A difficulty that we encounter with the L\'evy-Leblond equation (which we do not encounter while solving the Schr\"odinger equation) is that \( (\gamma_+)^2 = 0 \) and so \( \gamma_+ \) is not invertible.
This means that we are forced to leave \( \gamma_+ \) on the right-hand side of~\eqref{eq:TimeIndependentLevyLeblond}.

Note that \( \Hll \) is not self-adjoint and its eigenvalues do not correspond to the energy of the system.
Regardless, we will call \( \Hll \) the (L\'evy-Leblond) Hamiltonian due to the role it plays in our computations.
This is a non-standard approach; we normally require physical observables (such as energy) to be given by the spectra of self-adjoint operators.
Despite being non-standard,
our approach is consistent with the standard formulation of quantum mechanics: we will show (\Cref{prop:samespectrum}) that the energy levels \( E \) of the time-independent equation~\eqref{eq:TimeIndependentLevyLeblond} are contained in the spectrum of the familiar self-adjoint Schr\"odinger Hamiltonian.

Let \( \psi(\vec{x}) \) be a solution to the time-independent equation~\eqref{eq:TimeIndependentLevyLeblond}.
We will call \( E \) a \dfn{\( \gamma_+ \)-eigenvalue} of \( \Hll \) with \( \gamma_+ \)-eigenstate \( \psi(\vec{x}) \) and use associated terminology (such as \( \gamma_+ \)-eigenspaces).
If \( \psi(\vec{x}) \) is a solution to the time-independent equation~\eqref{eq:TimeIndependentLevyLeblond}, then
\[
	\Psi(t,\vec{x}) = e^{-iEt} \psi(\vec{x})
\]
is a solution to the time-dependent equation~\eqref{eq:TimeDependentLevyLeblond}.

\begin{rmk}
	Instead of~\eqref{eq:TimeIndependentLevyLeblond}, we could consider the more general time-independent equation
	\begin{equation}
		\Hll \psi(\vec{x}) = \gamma_+ (E + \gamma_+E_+ + \gamma_-E_- + \gamma^j E_j) \psi(\vec{x}).
	\end{equation}
	If we choose a matrix representation for the gamma matrices,
	the corresponding time-dependent solution is
	\[
		\Psi(t,\vec{x}) = \exp(-i(E + \gamma_+E_+ + \gamma_-E_- + \gamma^jE_j)t) \psi(\vec{x})
	\]
	where this \( \exp \) is a matrix exponential.
	For example, if we choose the representation~\eqref{eq:gammamatrixrep} then \( \exp(-i\gamma_-E_- t) = \Id - i \gamma_-E_-t \).
	However, the time-independent equation \eqref{eq:TimeIndependentLevyLeblond} is sufficient for our purposes.
\end{rmk}

If we apply \( \Hll \) to \eqref{eq:TimeIndependentLevyLeblond}, then the differential equation becomes
\begin{equation}\label{eq:TimeIndependentSchrodingerlike}
	\left(- \laplacian + \beta V(\vec{x}) - \gamma^j\gamma_+i\pdv{V}{x_j}\right)\psi(\vec{x}) = \beta E \psi(\vec{x})
\end{equation}
and so \( E \) is an eigenvalue of an operator that is similar to the Schr\"odinger Hamiltonian:
\begin{equation}
	\Hsch = -(1/\beta)\laplacian + V(\vec{x}) - \gamma^j\gamma_+ i \pdv{V}{x_j}(\vec{x})/\beta.
\end{equation}
That is, any solution to the time-independent L\'evy-Leblond equation~\eqref{eq:TimeIndependentLevyLeblond}
will be a solution to the time-independent Schr\"odinger-like equation of~\eqref{eq:TimeIndependentSchrodingerlike}.

In fact, the spectrum of \( \Hsch \) is contained in the spectrum of the familiar Schr\"odinger Hamiltonian \( -(1/\beta)\laplacian + V(x) \).
We illustrate this fact in the following example (and prove it more generally in \Cref{prop:samespectrum}).
\begin{ex}\label{ex:samespectrum}
	Consider the representation~\eqref{eq:gammamatrixrep} for the \( (1+1) \)-dimensional equation.
	In this case, the Schr\"odinger-like equation~\eqref{eq:TimeIndependentSchrodingerlike} becomes
	\[
		\Hsch 
		\begin{pmatrix}
			\psi_1(\vec{x})\\
			\psi_2(\vec{x})
		\end{pmatrix}
		= 
		\begin{pmatrix}
			\left(-\frac{1}{\beta}(\partial_1)^2 + V(\vec{x})\right)\psi_1(\vec{x})- \frac{i}{\beta}V'(\vec{x}) \psi_2(\vec{x})\\
			\left(-\frac{1}{\beta}(\partial_1)^2 + V(\vec{x})\right)\psi_2(\vec{x})
		\end{pmatrix}
		=
		\begin{pmatrix}
			E\psi_1(\vec{x})\\
			E\psi_2(\vec{x})
		\end{pmatrix}.
	\]
	But the second component is just the familiar time-independent Schr\"odinger equation.
	Therefore, \( E \) is either an eigenvalue of the Schr\"odinger Hamiltonian or \( \psi_2 = 0 \).
	But if \( \psi_2 = 0 \) then the first component becomes the familiar Schr\"odinger equation---in either case, \( E \) is an eigenvalue of the Schr\"odinger Hamiltonian.
\end{ex}

\subsection{Spectrum formalism}
In this paper, we will not concern ourselves with the full formality of unbounded operators acting on a Hilbert space.
However, given that we have introduced a new concept of \( \gamma_+ \)-eigenvalue, we will briefly explore this concept more formally in this section.

Since \( \Hll \) is an unbounded operator,
we should define the set of all \( \gamma_+ \)-eigenvalues \( E \) similar to a spectrum:
the \dfn{\( \gamma_+ \)-spectrum} of \( \Hll \) is the set of all \( E\in\CC \) such that \( \Hll - \gamma_+E \) has no everywhere-defined bounded inverse.
In general, a \( \gamma_+ \)-eigenvector corresponding to a member of the \( \gamma_+ \)-spectrum may lie outside the Hilbert space.

Observe that \( \Hll - \gamma_+E \) has bounded inverse if and only if \( (\Hll - \gamma_+ E)^2 = \Hll^2 - \beta E = \beta(\Hsch - E) \) has bounded inverse (provided \( \Hll \) is a closed operator), so the \( \gamma_+ \)-spectrum of \( \Hll \) is exactly the spectrum of \( \Hsch \).
In combination with the following proposition (c.f.\ \Cref{ex:samespectrum}), we can conclude that the \( \gamma_+ \)-spectrum of \( \Hll \) is contained in the spectrum of the corresponding Schr\"odinger Hamiltonian.

\begin{prop}\label{prop:samespectrum}
	Assume that \( H = -(1/\beta)\laplacian + V(x) \)
	is self-adjoint and that
	\(  \Hsch = H - \gamma^j\gamma_+ (i/\beta) \pdv{V}{x_j}(\vec{x}) \)
	is a closed operator with 
	the same domain as \( H \).
	Then, 
	the spectrum of \( \Hsch \) is contained in the spectrum of 
	\( H \).
\end{prop}
\begin{proof}
	Denote \( K = \ker \gamma_+ \).
	Since \( (\gamma_+)^2 = 0 \), we know that \( K \neq 0 \).
	Let \( \ocmpl{K} \) be the orthogonal complement of \( K \) so that the Hilbert space can be written as \( K\oplus\ocmpl{K} \).
	
	Consider an arbitrary state \( \ket{\psi} = \ket{\psi_K}\oplus\ket{\psi_\perp}\in K\oplus\ocmpl{K} \).
	Applying \( \Hsch - E\Id \) (for \( E\in\CC \) and \( \Id \) the identity operator), we find that
	\begin{align*}
		(\Hsch - E\Id)\ket{\psi} &= (H - E\Id)\ket{\psi} -  \gamma^j\frac{i}{\beta}\pdv{V}{x_j}\gamma_+\ket{\psi_K} + \gamma_+\gamma^j\frac{i}{\beta}\pdv{V}{x_j}\ket{\psi_\perp}\\
					 &= (H - E\Id)\ket{\psi} + \gamma_+\gamma^ji\pdv{V}{x_j}\ket{\psi_\perp}.
	\end{align*}

	The gamma matrix \( \gamma_+ \) commutes with \( H \),
	so \( \gamma_+(H\ket{\psi_K}) = 0 \) and \( H\ket{\psi_K} \in K \).
	Furthermore, \( H\ket{\psi_\perp}\in K^\perp \);
	indeed, we know that \( H \) is self-adjoint,
	so for any \( \ket{\varphi}\in K \) we have \( \bra{\varphi}(H\ket{\psi_\perp}) = \conj{(\bra{\psi_\perp}H\ket{\varphi})}=0 \) (because we have just shown that \( H\ket{\varphi}\in K \)).
	Finally, since \( (\gamma_+)^2 = 0 \) we have that \( \gamma_+\gamma^j(i/\beta)\pdv{V}{x_j}\in K \). 
	Thus,
	\[
		(\Hsch - E\Id)\ket{\psi} = \left((H-E\Id)\ket{\psi_K}+\gamma_+\gamma^j\frac{i}{\beta}\pdv{V}{x_j}\ket{\psi_\perp}\right) \oplus (H - E\Id)\ket{\psi_\perp}\in K\oplus\ocmpl{K}.
	\]

	Now, assume that \( H-EI \) has a bounded inverse.

	For injectivity,
	suppose \( (\Hsch - E\Id)\ket{\psi} = 0 \).
	Then \( (H-EI)\ket{\psi_{\perp}} = 0  \) and hence \( \ket{\psi_\perp} = 0 \).
	But since \( \ket{\psi_\perp} = 0 \), we have that \( (H-EI)\ket{\psi_K} = 0 \) and hence \( \ket{\psi_K}= 0 \).
	Thus, \( \ket{\psi} = 0 \) and \( \Hsch - E\Id \) is injective.

	For surjectivity, take an arbitrary \( \ket{\varphi} = \ket{\varphi_K}\oplus\ket{\varphi_\perp}\in K\oplus \ocmpl{K} \).
	By surjectivity of \( H-EI \), we can choose \( \ket{\psi_\perp}\in\ocmpl{K} \) such that 
	\( (H-EI)\ket{\psi_\perp} = \ket{\varphi_\perp} \)
	and then choose \( \ket{\psi_K}\in K \) such that \( (H-EI)\ket{\psi_K} = \ket{\varphi_K} - \gamma_+\gamma^j(i/\beta)\pdv{V}{x_j}\ket{\psi_\perp}  \).
	That is, \( (\Hsch-EI)\ket{\psi} = \ket{\varphi} \) and \( \Hsch - EI \) is surjective.

	Finally, \( \Hsch - EI \) is closed by assumption, and so its inverse is bounded by the Closed Graph Theorem.

	Taking the contrapositive, if \( E \) is in the spectrum of \( \Hsch \) then \( E \) is in the spectrum of \( H \).
\end{proof}

If we take the representation in~\eqref{eq:gammamatrixrep} with a free or harmonic potential,
then we can easily verify that \( H \) is self-adjoint with the same domain as \( \Hsch \).
Moreover, observe that
\[
	\gamma_+\gamma^j\frac{i}{\beta}\pdv{V}{x_j} = -\gamma_+\frac{i}{\beta}V'(x_1) = \adjt{\left(\gamma_-\frac{i}{\beta}\conj{(V'(x_1))}\right)}
\]
and so \( \Hsch \) is closed in this representation, being the adjoint of \( H+\gamma_-(i/\beta)\conj{(V'(x_1))} \)
(the adjoint of an operator is always closed).
Thus, we are justified in applying \Cref{prop:samespectrum} to the operators used in this paper.

\subsection{Colour Lie superalgebras}
The \( \Ztzt \)-graded colour Lie superalgebras generalise Lie superalgebras to grading by the group \( \Ztzt \).
Within this paper,
every \dfn{\( \Ztzt \)-graded colour Lie superalgebra}
appears as a \( \Ztzt \)-graded space of operators \( \g = \g_{00}\oplus\g_{01}\oplus\g_{10}\oplus\g_{11} \) acting on a Hilbert space of states.
This space \( \g \) is equipped with a so-called colour bracket \( \cbrak{\cdot}{\cdot}\colon\g\times\g\to\g \) realised as
\[
	\cbrak{x}{y} = xy - \varepsilon(\alpha,\beta)yx \qquad \text{for}\;x\in\g_{\alpha},\,y\in\g_{\beta}\;\text{and satisfying}\;\cbrak{x}{y}\in\g_{\alpha+\beta}
\]
with \dfn{commutation factor} \( \varepsilon(\alpha,\beta) = (-1)^{\alpha_1\beta_1 + \alpha_2\beta_2} \) where \( \alpha=(\alpha_1,\alpha_2),\,\beta=(\beta_1,\beta_2)\in\Ztzt \).
By replacing the commutation factor \( \varepsilon \) with \( \varepsilon(\alpha,\beta) = (-1)^{\alpha_1\beta_2 - \alpha_2\beta_1} \), we obtain what is called a \dfn{\( \Ztzt \)-graded colour Lie algebra} (as opposed to a \emph{super}algebra; however, despite the names, colour Lie (super)algebras are neither Lie algebras nor Lie superalgebras).
The \( \Ztwo[n] \)-graded colour Lie superalebras are the natural generalisation of \( \Ztzt \)-graded colour Lie superalgebras.
The above definitions are sufficient for this paper;
we direct the interested reader to the more general definitions of colour Lie algebras in~\cite{RW1978a,RW1978b,Scheunert1979}.

In~\cite{AKTT2016}, the authors performed an exhaustive search of the first-order differential operators which satisfied
\[
	\comm{\Omega}{Z}=\Phi_Z(\vec{x})\Omega \qquad \text{or} \qquad \acomm{\Omega}{Z}=\Phi_Z(\vec{x})\Omega 
\]
where \( \Omega \) is the L\'evy-Leblond operator and \( \Phi_Z \) is a matrix-valued function of the space coordinates.
By investigating the algebras generated by these operators, the authors found a \( \Ztzt \)-graded colour Lie superalgebra as an extension of the Schr\"odinger symmetry algebra.
However, it was proved that there is no \( \Ztzt \)-graded extension of the Schr\"odinger symmetry algebra for the L\'evy-Leblond equation with harmonic potential~\cite{AKTT2016}.

\section{Eigenspace-preserving operators of the free equation}\label{sec:free}
Inspired by the symmetry algebras discovered in~\cite{AKTT2016,AKTT2017}, we wish to identify the operators which leave the \( \gamma_+ \)-eigenspaces invariant,
and determine if any \( \Ztzt \)-graded colour Lie (super)algebras appear.
We will work in natural units, \( \hbar = 1 \).

Recall from~\eqref{eq:TimeIndependentLevyLeblond} that the time-independent \( (1+1) \)-dimensional L\'evy-Leblond equation (for the Schr\"odinger equation) with free potential is
\begin{equation}\label{eq:1+1TimeIndependentLevyLeblond}
	(\gamma_- \beta - \gamma^1 i \partial_1) \psi(x_1) = \gamma_+ E \psi(x_1)
\end{equation}
where \( t \) is the time coordinate and \( x_1 \) is the first (and only) spacial coordinate.
For ease of reading, we will omit the subscript and write \( x \) as the spacial coordinate.

By \Cref{prop:samespectrum},
\( E \) is in the spectrum of the free Schr\"odinger Hamiltonian \( -(1/\beta)(\partial_1)^2 \).
Since the free Schr\"odinger Hamiltonian is proportional to the square of the self-adjoint momentum operator \( -i\partial_1 \),
we have that \( E \in [0,\infty) \).

\subsection{Eigenvalue conditions}
Let \( \ket{\psi} \equiv \psi(x) \) be a solution to the time-independent equation~\eqref{eq:1+1TimeIndependentLevyLeblond}.
We wish to find an operator \( A \) that leaves the \( \gamma_+ \)-eigenspaces invariant;
that is,
\begin{equation}\label{eq:EigspaceInvariant}
	(\gamma_-\beta - \gamma^1 i \partial_1)A\ket{\psi} = \gamma_+ E A\ket{\psi}
\end{equation}
for \( \ket{\psi} \) a solution to the time-independent equation~\eqref{eq:1+1TimeIndependentLevyLeblond}.

Note that we can rearrange equation~\eqref{eq:1+1TimeIndependentLevyLeblond} to solve for \( \partial_1 \):
\begin{equation}\label{eq:DifferentialConsequence}
	\begin{aligned}
		\partial_1 \ket{\psi} &= (\gamma^1\gamma_+ iE  - \gamma^1\gamma_-i\beta)\ket{\psi}\\
				      &= (\gamma_+ iE + \gamma_-i\beta)\ket{\psi}
	\end{aligned}
\end{equation}
using relations~\eqref{eq:gammaMatricesFreeEquation}.
That is,
the operator \( \partial_1 \) can be expressed in terms of gamma matrices when acting on a \( \gamma_+ \)-eigenspace.
Therefore, 
if we only consider the action on a \( \gamma_+ \)-eigenspace,
\( A \) does not need to contain any differential operators
(this fact is unique to the \( (1+1) \)-dimensional case).
So the general form of \( A \) (restricted to a \( \gamma_+ \)-eigenspace) is
\begin{equation}\label{eq:1+1InvariantOperatorGeneralForm}
	A = c_{\Id}(x)\Id + c_+(x) \gamma_+ + c_-(x) \gamma_- + c_1(x)\gamma^1
\end{equation}
for some complex-valued functions \( c_{\Id},\, c_+,\, c_-,\, c_1 \).
We will assume that these functions are differentiable.

We then substitute the general \( A \) of~\eqref{eq:1+1InvariantOperatorGeneralForm} into~\eqref{eq:EigspaceInvariant},
and use relations~\eqref{eq:gammaMatricesFreeEquation}, \eqref{eq:DifferentialConsequence}
and the product rule: \( \partial_1 f(x) = f(x) \partial_1 + f'(x) \).
In doing so, we obtain a coupled ODE:
\begin{equation}\label{eq:1+1ODE}
	\left\{
		\begin{aligned}
			\dv{c_{\Id}}{x} &= 0 \\
			\dv{c_+}{x} &= -2iE c_1(x)\\
			\dv{c_-}{x} &= 2i\beta c_1(x)\\
			\dv{c_1}{x} &= -i\beta c_+(x) + iE c_-(x).
		\end{aligned}
	\right.
\end{equation}
Obviously \( c_{\Id}(x) \equiv c_{\Id} \) is a constant.
Defining a new variable \( c(x) = -i\beta c_+(x) + iEc_-(x) \), yields the following coupled ODE:
\[
	\left\{
		\begin{aligned}
			\dv{c}{x} &= -4E\beta c_1(x) \\
			\dv{c_1}{x} &= c(x).
		\end{aligned}
	\right.
\]
which we can solve to find that
\begin{align*}
	c_1(x) &= ae^{i2\sqrt{E\beta}x} + be^{-i2\sqrt{E\beta}x} \\
	c(x) &= i2\sqrt{E\beta}(ae^{i2\sqrt{E\beta}x} - be^{-i2\sqrt{E\beta}x})\\
	c_+(x) &= -a \sqrt{\tfrac{E}{\beta}} e^{i2\sqrt{E\beta}x} + b \sqrt{\tfrac{E}{\beta}} e^{-i2\sqrt{E\beta}x} + \frac{d}{\beta}\\
	c_-(x) &= a \sqrt{\tfrac{\beta}{E}} e^{i2\sqrt{E\beta}x} - b \sqrt{\tfrac{\beta}{E}} e^{-i2\sqrt{E\beta}x} + \frac{d}{E}
\end{align*}
for some constants \( a,\, b,\, d \) and
assuming \( E \) and \( \beta \) are non-zero.
Substituting the expressions we have computed for \( c_I,\, c_+,\, c_-,\, c_1 \) back into the general form of \( A \) in \eqref{eq:1+1InvariantOperatorGeneralForm}, we find four linearly independent operators:
\begin{gather}
	I \label{eq:eigvalopI}\\
	\gamma_+ iE + \gamma_- i\beta \label{eq:eigvalopD}\\
	-\gamma_+ \sqrt{\frac{E}{\beta}} e^{i2\sqrt{E\beta} x} + \gamma_- \sqrt{\frac{\beta}{E}} e^{i2\sqrt{E\beta} x} + \gamma^1 e^{i2\sqrt{E\beta} x} \label{eq:eigvalopR1}\\
	\gamma_+ \sqrt{\frac{E}{\beta}} e^{-i2\sqrt{E\beta} x} - \gamma_- \sqrt{\frac{\beta}{E}} e^{-i2\sqrt{E\beta} x} + \gamma^1 e^{-i2\sqrt{E\beta} x} \label{eq:eigvalopR2}.
\end{gather}
We can easily verify that these operators leave the \( \gamma_+ \)-eigenspaces invariant.

For a specific choice of \( E \), the operators~\labelcref{eq:eigvalopI,eq:eigvalopD,eq:eigvalopR1,eq:eigvalopR2} are only defined on the corresponding \( \gamma_+ \)-eigenspace.
If we want operators defined everywhere on (a dense subset of) the Hilbert space, then
 we can interpret~\labelcref{eq:eigvalopI,eq:eigvalopD,eq:eigvalopR1,eq:eigvalopR2} as eigenvalue conditions which the desired operators must satisfy.

\subsection{Finding operators}
From~\eqref{eq:DifferentialConsequence} we immediately see that the operator \( \partial_1 \) has the same action as~\eqref{eq:eigvalopD} on the \( \gamma_+ \)-eigenspace.
We will multiply through by \( -i \) and instead consider the self-adjoint momentum operator \( -i\partial_1 \).
Since \( -i\partial_1 \) commutes with the L\'evy-Leblond Hamiltonian \( \Hll = \gamma_-\beta - \gamma^1 i \partial_1 \) and with \( \gamma_+ E \),
we can easily verify that 
\( -i\partial_1 \) does indeed leave the \( \gamma_+ \)-eigenspaces invariant.
Surprisingly, \eqref{eq:eigvalopD} also gives rise to another operator.
If we multiply~\eqref{eq:eigvalopD} by \( -i\beta \) and use the fact that \( E \) is an eigenvalue of the Schr\"odinger Hamiltonian, then we obtain the operator \( -\gamma_+(\partial_1)^2 + \gamma_- \beta^2 \).
This operator does leave the \( \gamma_+ \)-eigenspaces invariant:
\begin{align*}
	&\Hll(-\gamma_+(\partial_1)^2 + \gamma_- \beta^2)\ket{\psi}\\
	&= -(-\gamma_+(\partial_1)^2 + \gamma_- \beta^2)\Hll\ket{\psi} + \acomm{\Hll}{-\gamma_+(\partial_1)^2 + \gamma_- \beta^2}\ket{\psi}\\
	&= -(-\gamma_+(\partial_1)^2 + \gamma_- \beta^2)\gamma_+E\ket{\psi} + I\beta^2\left(-\frac{1}{\beta}(\partial_1)^2\right)\ket{\psi}\\
	&= \gamma_+E(-\gamma_+(\partial_1)^2 + \gamma_- \beta^2)\ket{\psi} - \acomm{-\gamma_+(\partial_1)^2 + \gamma_- \beta^2}{\gamma_+E}\ket{\psi} + I\beta^2E\ket{\psi}\\
	&= \gamma_+E(-\gamma_+(\partial_1)^2 + \gamma_- \beta^2)\ket{\psi} - I\beta^2E\ket{\psi} + I\beta^2E\ket{\psi}\\
	&= \gamma_+E(-\gamma_+(\partial_1)^2 + \gamma_- \beta^2)\ket{\psi}.
\end{align*}

Now, we shall examine \eqref{eq:eigvalopR2}.
We need to find an operator which has \( \ket{\psi} \) as an eigenstate with eigenvalue involving \( \sqrt{E} \).
We would expect the eigenvalues of \( (\partial_1)^2 \) to be the squares of the eigenvalues of \( \partial_1 \); that is, we expect the eigenvalues of \( \partial_1 \) to be \( \pm i\sqrt{E\beta} \).
For now, we shall assume that \( \ket{\psi} \) is an eigenstate of \( \partial_1 \) with eigenvalue \( i\sqrt{E\beta} \).
With this eigenvalue, we can find an operator corresponding to \( e^{-2i\sqrt{E\beta} x} \):
\begin{align*}
	e^{-2i\sqrt{E\beta} x} \ket{\psi} 
	&= \sum_{k=0}^{\infty} \frac{(-2x)^k (i\sqrt{E\beta})^k \psi(x)}{k!} && \text{where}~\psi(x)\equiv\ket{\psi}\\
	&= \sum_{k=0}^{\infty} \frac{(-2x)^k (\partial_1)^k \psi(x)}{k!} \\
	&= \sum_{k=0}^{\infty} \frac{\psi^{(k)}(x)}{k!} (-x - x)^k.
\end{align*}
This final expression is the Taylor series for \( \psi(-x) \) about the point \( x \).
If we assume that \( \psi \) is analytic everywhere, then \( e^{-2i\sqrt{E\beta} x}\ket{\psi} = \psi(-x) \).
So, on an analytic function, \( e^{-2i\sqrt{E\beta}} \) acts the same as the parity operator \( \parity \),
defined by \( \parity (\chi(x)) = \chi(-x) \) for all states \( \chi \).
Note that \( \partial_1 \parity = - \parity \partial_1 \) by the chain rule.

Assuming that \( \ket{\psi} \) is an analytic function, from~\eqref{eq:eigvalopR2} we compute:
\begin{align*}
	&\left(\gamma_+ \sqrt{\frac{E}{\beta}} e^{-i2\sqrt{E\beta} x} - \gamma_- \sqrt{\frac{\beta}{E}} e^{-i2\sqrt{E\beta} x} + \gamma^1 e^{-i2\sqrt{E\beta} x}\right)\ket{\psi}\\
	&= \left(\gamma_+ \frac{iE}{i\sqrt{E\beta}} \parity - \gamma_- \frac{i\beta}{i\sqrt{E\beta}} \parity + \gamma^1 \parity\right)\ket{\psi} \\
	&= \parity\left(\gamma^1 \frac{1}{i\sqrt{E\beta}}\partial_1 + \gamma^1\right)\ket{\psi} \\
	&= 2\gamma^1\parity\ket{\psi}
\end{align*}
using~\eqref{eq:DifferentialConsequence} and the fact that \( \partial_1 \) has eigenvalue \( i\sqrt{E\beta} \).
Therefore, we guess that \( \gamma^1\parity \) is an operator which leaves the \( \gamma_+ \)-eigenspace invariant.
And indeed, we can verify that this is true for all eigenstates \( \ket{\psi} \) (not just analytic ones):
\begin{align*}
	\Hll(\gamma^1\parity \ket{\psi})
	&= \gamma_-\beta\parity \ket{\psi} - i\partial_1\parity \ket{\psi} \\
	&= \gamma_-\beta\parity \ket{\psi} + i\parity \partial_1\ket{\psi} \\
	&= \gamma_-\beta\parity \ket{\psi} + i\parity (\gamma_+iE + \gamma_-i\beta)\ket{\psi} \\
	&= -\gamma_+E \parity \ket{\psi} \\
	&= \gamma_+E (\gamma^1\parity \ket{\psi}).
\end{align*}
Note that we did not use the assumption that the eigenvalue of \( \partial_1 \) is \( i\sqrt{E\beta} \) in the above computation.
 
To find \( \gamma^1\parity \), we assumed that \( \partial_1 \) has eigenvalue \( i\sqrt{E\beta} \).
If we instead assume that \( \partial_1 \) has eigenvalue \( -i\sqrt{E\beta} \),
then the operator we obtain from~\eqref{eq:eigvalopR2} is \( 0 \), since 
\begin{align*}
	&\left(\gamma_+ \sqrt{\frac{E}{\beta}} e^{-i2\sqrt{E\beta} x} - \gamma_- \sqrt{\frac{\beta}{E}} e^{-i2\sqrt{E\beta} x} + \gamma^1 e^{-i2\sqrt{E\beta} x}\right)\ket{\psi}\\
	&= e^{-i2\sqrt{E\beta} x}\left(\gamma_1 \frac{1}{i\sqrt{E\beta}}\partial_1 + \gamma^1\right)\ket{\psi}.
\end{align*}
The behaviour for~\eqref{eq:eigvalopR1} is similar:
if \( \partial_1 \) has eigenvalue \( -i\sqrt{E\beta} \) then we again obtain the operator \( \gamma^1\parity \) and if \( \partial_1 \) has eigenvalue \( i\sqrt{E\beta} \) then we obtain the zero operator.

So far, we have identified three operators which leave the \( \gamma_+ \)-eigenspace invariant:
\[
	\partial_1, \qquad -\gamma_+(\partial_1)^2 + \gamma_- \beta^2,\qquad \gamma^1\parity.
\]
This is not an exhaustive list of such operators.
Indeed, composing any number of these three operators will yield further operators which leave the \( \gamma_+ \)-eigenspaces invariant.
However, we will not consider any further \( \gamma_+ \)-eigenspace-preserving operators in this paper.


\subsection{Grading the operators}
We have found operators which leave the \( \gamma_+ \)-eigenspaces invariant.
We now wish to determine if these operators close to form an algebra.

Let,
\[
	\Hsch = -\frac{1}{\beta}(\partial_1)^2,
	\qquad\widehat{P} = -i\partial_1,
	\qquad D_+ = -\gamma_+\frac{1}{\beta}(\partial_1)^2 + \gamma_- \beta,
	\qquad \parity^1 = \gamma^1\parity.
\]
Computing the complete set of commutation and anticommutation relations we find
\begin{equation}\label{eq:1+1InvariantRelations}
	\begin{aligned}
		\comm*{\widehat{P}}{\widehat{P}} &= 0 & 
		\acomm*{\widehat{P}}{\widehat{P}} &= 2\beta\Hsch\\
		\comm*{D_+}{D_+} &= 0 & 
		\acomm*{D_+}{D_+} &= 2\beta\Hsch \\
		\comm*{\parity^1}{\parity^1} &= 0 & 
		\acomm*{\parity^1}{\parity^1} &= 2\Id\\
		\comm*{\widehat{P}}{D_+} &= 0 &
		\acomm*{\widehat{P}}{D_+} &= 2D_+\widehat{P} \\
		\comm*{\widehat{P}}{\parity^1} &= -2\parity^1\widehat{P} &
		\acomm*{\widehat{P}}{\parity^1} &= 0 \\
		\comm*{D_+}{\parity^1} &= -2\parity^1D_+ & 
		\acomm*{D_+}{\parity^1} &= 0\\
	\end{aligned}
\end{equation}
and \( \Hsch \) commutes with the other three operators.

We can create colour Lie (super)algebras by assigning gradings to these operators to choose which of the above (anti)commutation relations are expressed.

\begin{thm}
	Let \( \mathfrak{A},\, \mathfrak{D}^+,\, \mathfrak{D}^1,\, \mathfrak{D} \) be vector spaces.
	Define \( \Ztzt \)-graded sectors for \( \mathfrak{A},\, \mathfrak{D}^+,\, \mathfrak{D}^1 \)
	and \( \Ztwo[3] \)-graded sectors for \( \mathfrak{D} \) as follows:
	\begin{align*}
		\mathfrak{A}_{00} &= 0, & \mathfrak{A}_{01} &= \spn\{\widehat{P},D_+\}, & \mathfrak{A}_{10} &= \spn\{\parity^1\}, & \mathfrak{A}_{11} &= 0,\\[1ex]
		\mathfrak{D}^+_{00} &= \spn\{\Id,\Hsch\}, & \mathfrak{D}^+_{01}&=\spn\{\widehat{P}\}, & \mathfrak{D}^+_{10}&=\spn\{D_+\}, & \mathfrak{D}^+_{11} &= \spn\{\parity^1\},\\[1ex]
		\mathfrak{D}^1_{00} &= \spn\{\Id\}, & \mathfrak{D}^1_{01}&=\spn\{\parity^1\}, & \mathfrak{D}^1_{10}&=0, & \mathfrak{D}^1_{11} &= \spn\{\widehat{P},D_+\},\\[1ex]
		\mathfrak{D}_{000} &= \spn\{\Id,\Hsch\}, & \mathfrak{D}_{001} &= \spn\{\widehat{P}\}, & \mathfrak{D}_{010} &= \spn\{D_+\}, & \mathfrak{D}_{011} &= 0, \\[-1ex]
		\mathfrak{D}_{100}&= 0, & \mathfrak{D}_{101}&=0, & \mathfrak{D}_{110}&=0, & \mathfrak{D}_{111} &= \spn\{\parity^1\}.
	\end{align*}
	We then define \( \mathfrak{A},\, \mathfrak{D}^+,\, \mathfrak{D}^1,\, \mathfrak{D} \) to be the direct sum of their respective sectors.
	Define a bracket on each space by \( \cbrak{A}{B} = AB - \varepsilon(\alpha,\beta) BA \) for \( A,\, B \) homogeneous of degree \( \alpha,\,\beta \) respectively.
	Here, \( \varepsilon \) is the commutation factor for the algebra
	\begin{align*}
		\text{for}~\mathfrak{A}\colon&& \varepsilon(\alpha_1\alpha_2,\beta_1\beta_2) &= (-1)^{\alpha_1\cdot \beta_2 - \alpha_2\cdot \beta_1} \\
		\text{for}~\mathfrak{D}^+~\text{and}~\mathfrak{D}^1\colon&& \varepsilon(\alpha_1\alpha_2,\beta_1\beta_2) &= (-1)^{\alpha_1\cdot \beta_1 + \alpha_2\cdot \beta_2} \\
		\text{for}~\mathfrak{D}\colon && \varepsilon(\alpha_1\alpha_2\alpha_3,\beta_1\beta_2\beta_3) &= (-1)^{\alpha_1\cdot\beta_1 + \alpha_2\cdot\beta_2 + \alpha_3\cdot\beta_3}.
	\end{align*}
	Then \( \mathfrak{D}^+ \) closes to form a colour Lie algebra and \( \mathfrak{A} \), \( \mathfrak{D}^- \), and \( \mathfrak{D} \) close to form colour Lie superalgebras.
\end{thm}

\begin{proof}
	By examining~\eqref{eq:1+1InvariantRelations}, we verify that the colour bracket of the \( \alpha \)-sector and the \( \beta \)-sector is a subspace of the \( (\alpha+\beta) \)-sector (for all possible sectors).
\end{proof}

If we want to construct an algebra that contains only \( \widehat{P},\,D_+,\,\parity^1 \),
then we need to define a bracket \( \cbrak{\cdot}{\cdot} \) such that \( \cbrak*{\widehat{P}}{\widehat{P}} = \comm*{\widehat{P}}{\widehat{P}} \)
but \( \cbrak*{\widehat{P}}{\parity^1} = \acomm*{\widehat{P}}{\parity^1} \).
This is not possible with a Lie superalgebra, but is possible with a \( \Ztzt \)-graded colour Lie (super)algebra such as \( \mathfrak{A} \). However, \( \mathfrak{A} \) is somewhat trivial because \( \cbrak{A}{B} = 0 \) for all \( A,B \in \mathfrak{A} \).

The most interesting relations of~\eqref{eq:1+1InvariantRelations} are \( \acomm{D_+}{D_+} = \beta\Hsch \) and \( \acomm{\parity^1}{\parity^1} = 2I \).
We would hope to capture these relations in the colour Lie (super)algebra.
Two options for a \( \Ztzt \)-graded algebra
which realise \( \acomm{D_+}{D_+} \) or \( \acomm{\parity^1}{\parity^1} \) are \( \mathfrak{D}^+ \) and \( \mathfrak{D}^1 \) (respectively).

It is not possible to realise both \( \acomm{D_+}{D_+} \) and \( \acomm{\parity^1}{\parity^1} \) simultaneously with a \( \Ztzt \)-grading on at most five basis elements.
We can, however, do this with the \( \Ztwo[3] \)-graded algebra \( \mathfrak{D} \).

By considering only symmetry operators that are first- or second-order differential operators,
\cite{AKTT2016,AKTT2017} found a \( \Ztzt \)-graded colour Lie superalgebra.
But observe that \( \parity^1 \) is not a differential operator (though is a symmetry operator).
Thus, the appearance of a \( \Ztwo[3] \)-grade in the above theorem suggests that a larger grading is required when considering a broader class of symmetry operators.
However, we are primarily concerned with the utility of \( \mathfrak{D} \) for finding solutions; a more in-depth analysis for the symmetry operators of the free equation is beyond the scope of this paper.

\subsection{Relationship between the algebra and solutions}\label{sec:solvefreeequation}
The \( \Ztwo[3] \)-graded algebra \( \mathfrak{D} \) contains two operators, \( \widehat{P} \) and \( D_+ \), which both square to the Schr\"odinger Hamiltonian \( \Hsch \).
Finding the eigenstates of either operator would then allow us to solve the Schr\"odinger equation.
We will prove that solving the L\'evy-Leblond equation is equivalent to finding \emph{simultaneous} eigenstates of these operators.
In a sense, this couples together these ``square roots'' of \( \Hsch \) and leads to the coupling of the components in the solutions to the L\'evy-Leblond equation.

\begin{thm}\label{thm:simulstates}
	Solving the time-independent L\'evy-Leblond equation~\eqref{eq:1+1TimeIndependentLevyLeblond} with \( E>0 \)
	is equivalent to
	finding the simultaneous eigenstates for \( \widehat{P} \) and \( D_+ \) with positive eigenvalues for both operators.
	Specifically,  \( \ket{\psi} \) is a solution to the time independent L\'evy-Leblond equation
	with \( E>0 \) if and only if
	\[
		\ket{\psi} = a\ket{\varphi_1} + b\parity^1\ket{\varphi_2} \qquad a,b\in\CC 
	\]
	for some simultaneous eigenstates \( \ket{\varphi_1},\, \ket{\varphi_2} \) of \( \widehat{P} \) and \( D_+ \) with all eigenvalues positive.
\end{thm}
\begin{proof}
	For the forward direction, let \( \ket{\psi} \) be a solution to the time-independent L\'evy-Leblond equation.
	Let \( \ket{\chi} = (1/\sqrt{E\beta})\widehat{P}\ket{\psi} \).

	First, consider the case \( \ket{\chi} \neq c\ket{\psi} \) for any \( c\in\CC \).
	Set \( \ket{\varphi_1} = \ket{\psi} + \ket{\chi} \) and \( \ket{\varphi_2} = \parity^1(\ket{\psi} - \ket{\chi}) \). 
	Since \( \Hsch = (1/\beta)\widehat{P}^2 \),
	we know that \( (\widehat{P})^2\ket{\psi} = E\beta\ket{\psi} \),
	and so \( \widehat{P}\ket{\chi} = \sqrt{E\beta}\ket{\psi} \).
	Noting also that \( \widehat{P}\ket{\psi} = \sqrt{E\beta}\ket{\chi} \),
	it is easily verified that
	\[
		\widehat{P}\ket{\varphi_1} = \sqrt{E\beta}\ket{\varphi_1}
		\qquad \text{and} \qquad
		\widehat{P}\ket{\varphi_2} = \sqrt{E\beta}\ket{\varphi_2}.
	\]
	From~\eqref{eq:DifferentialConsequence},
	we know that 
	\[
		\widehat{P}\ket{\psi} = (\gamma_+ E + \gamma_- \beta)\ket{\psi} = \left(-\gamma_+ \frac{1}{\beta}(\partial_1)^2 + \gamma_- \beta\right)\ket{\psi} = D_+\ket{\psi}.
	\]
	Therefore, \( \ket{\chi} = (1/\sqrt{E\beta})D_+\ket{\psi} \) and a similar argument shows that
	\[
		D_+\ket{\varphi_1} = \sqrt{E\beta}\ket{\varphi_1}
		\qquad \text{and} \qquad
		D_+\ket{\varphi_2} = \sqrt{E\beta}\ket{\varphi_2}.
	\]
	That is, \( \ket{\varphi_1} \) and \( \ket{\varphi_2} \) are simultaneous eigenstates for \( \widehat{P} \) and \( D_+ \) with all eigenvalues positive.
	Moreover, \( \ket{\psi} = \ket{\varphi_1} + \parity^1\ket{\varphi_2} \) so we set \( a=b=1 \).

	Now, consider the case where \( \ket{\chi} = c \ket{\psi} \) for some \( c\in\CC \). Then \( \ket{\psi} \) is an eigenstate for \( \widehat{P} \). 
	From~\eqref{eq:DifferentialConsequence},
	we know that 
	\[
		\widehat{P}\ket{\psi} = (\gamma_+ E + \gamma_- \beta)\ket{\psi} = (-\gamma_+ \frac{1}{\beta}(\partial_1)^2 + \gamma_- \beta)\ket{\psi} = D_+\ket{\psi}.
	\]
	Therefore, \( \ket{\psi} \) is also an eigenstate for \( D_+ \) with the same eigenvalue.
	Since \( \Hsch = (1/\beta) (\widehat{P})^2 \), we know that \( (\widehat{P})^2\ket{\psi} = E\beta\ket{\psi} \), and so \( c^2 = 1 \).
	We have thus shown that \( \ket{\psi} \) is a simultaneous eigenstate of \( \widehat{P} \) and \( D_+ \)
	with identical non-zero eigenvalues.
	Expressing \( \ket{\psi} \) as a trivial linear combination \( \ket{\psi}=\ket{\psi} \)
	(if the eigenvalues of \( \ket{\psi} \) are positive) 
	or \( \ket{\psi} = \parity^1(\parity^1\ket{\psi}) \)
	(if the eigenvalues of \( \ket{\psi} \) are negative) proves the theorem in this case.
	More explicitly, if \( c = 1 \) then we can choose \( \ket{\varphi_1} = \ket{\psi} \) with \( a=1,\, b=0 \).
	If \( c = -1 \) then \( \parity^1\ket{\psi} \) has positive eigenvalues
	(since \( \acomm*{\widehat{P}}{\parity^1} = \acomm*{D_+}{\parity^1} = 0 \))
	so we can choose \( \ket{\varphi_2} = \parity^1\ket{\psi} \) with \( a=0,\, b=1 \).

	For the reverse direction, let \( \ket{\varphi_1} \) and \( \ket{\varphi_2} \) be simultaneous eigenstates of \( \widehat{P} \) and \( D_+ \) with only positive eigenvalues.
	Since both operators square to \( \beta\Hsch \) we have that \( \ket{\varphi_1} \) and \( \ket{\varphi_2} \) are also eigenstates of \( \Hsch \).
	Let \( E>0 \) be the corresponding eigenvalue of \( \Hsch \).
	Then the eigenvalues for \( \widehat{P} \) and \( D_+ \) must be \( \sqrt{E\beta} \)
	(note that we are assuming only positive eigenvalues).
	In particular,
	\[
		-i\partial_1 \ket{\varphi_1} = \sqrt{E\beta}\ket{\varphi_1} =  \left(-\gamma_+ \frac{1}{\beta}(\partial_1)^2 + \gamma_-\beta\right)\ket{\varphi_1} = (\gamma_+ E + \gamma_- \beta)\ket{\varphi_1}.
	\]
	Rearranging, we get
	\[
		(\gamma_-\beta - \gamma^1 i \partial_1)\ket{\varphi_1} = \gamma_+ E\ket{\varphi_1}
	\]
	so \( \ket{\varphi_1} \) is a solution to the time-independent L\'evy-Leblond equation.
	Note that the eigenvalue of \( \parity^1\ket{\varphi_2} \) with \( \widehat{P} \) and \( D_+ \) is \( -\sqrt{E\beta} \), so a similar argument shows that \( \parity^1\ket{\varphi_2} \) is a solution to the L\'evy-Leblond equation.
	Since all the operators are linear, any linear combination of \( \ket{\varphi_1} \) and \( \parity^1\ket{\varphi_2} \) will also be a solution to the L\'evy-Leblond equation.
\end{proof}

\Cref{thm:simulstates} further demonstrates that the operators of the \( \Ztzt \)-graded algebra \( \mathfrak{D} \) are fundamental to the free L\'evy-Leblond equation.

\subsection{Solving the free equation}
To conclude this section, we will use \Cref{thm:simulstates} to solve the L\'evy--Leblond equation. 
This is mostly useful as an illustration; the \( (1+1) \)-dimensional L\'evy-Leblond equation is easy enough to solve without (directly) appealing to this theorem.
So far, our results have been representation independent (assuming relations~\eqref{eq:gammaMatricesFreeEquation} hold).
However, to find a solution we will choose the matrix representation~\eqref{eq:gammamatrixrep}.

We are looking for simultaneous eigenstates of \( \widehat{P}=-i\pdv{x} \) and \( D_+ = -\gamma_+(1/\beta)\pdv[2]{x} + \gamma_- \beta \) with only positive eigenvalues.
We will pre-emptively look for eigenstates with eigenvalue \( \sqrt{E\beta} \).
We know that, up to scaling, the only eigenfunction of the derivative operator with eigenvalue \( i\sqrt{E\beta} \) is the exponential function \( \exp(i\sqrt{E\beta}x) \).
Therefore, the eigenstates for \( \widehat{P} \) are of the form
\[
	\begin{pmatrix}
		C_1 e^{i\sqrt{E\beta} x}\\
		C_2 e^{i\sqrt{E\beta} x}\\
	\end{pmatrix}
\]
for some constants \( C_1,\, C_2 \in \CC \).
Now, we examine which of these eigenstates are also eigenstates of \( D_+ \):
\begin{align*}
	D_+
	\begin{pmatrix}
		C_1 e^{i\sqrt{E\beta} x}\\
		C_2 e^{i\sqrt{E\beta} x}\\
	\end{pmatrix}
	&=
	\begin{pmatrix}
		C_2 E e^{i\sqrt{E\beta} x}\\
		C_1 \beta e^{i\sqrt{E\beta} x}\\
	\end{pmatrix}
	=
	\frac{C_1\beta}{C_2}
	\begin{pmatrix}
		\frac{(C_2)^2 E}{(C_1)^2 \beta} C_1e^{i\sqrt{E\beta} x}\\
		 C_2 e^{i\sqrt{E\beta} x}\\
	\end{pmatrix}
\end{align*}
so we must have \( (C_1)^2\beta = (C_2)^2 E \).
Therefore, up to scaling, the only simultaneous eigenstate with only positive eigenvalues is
\[
	\begin{pmatrix}
		\sqrt{E}e^{i\sqrt{E\beta} x}\\
		\sqrt{\beta}e^{i\sqrt{E\beta} x}
	\end{pmatrix}.
\]
By \Cref{thm:simulstates},
the general solution for \( E>0 \) is 
\[
	a
	\begin{pmatrix}
		\sqrt{E}e^{i\sqrt{E\beta} x}\\
		\sqrt{\beta}e^{i\sqrt{E\beta} x}
	\end{pmatrix}
	+
	b
	\gamma^1\parity
	\begin{pmatrix}
		\sqrt{E}e^{i\sqrt{E\beta} x}\\
		\sqrt{\beta}e^{i\sqrt{E\beta} x}
	\end{pmatrix}
	=
	\begin{pmatrix}
		a\sqrt{E}e^{i\sqrt{E\beta} x} + b \sqrt{E}e^{-i\sqrt{E\beta} x}\\
		a\sqrt{\beta}e^{i\sqrt{E\beta} x} - b \sqrt{\beta}e^{-i\sqrt{E\beta} x}
	\end{pmatrix}
\]
for any \( a,\, b\in\CC \).
Note that if we substitute \( E = 0 \), then the above is still an eigenvector.
For every \( k\in\RR \), set
\begin{equation}\label{eq:1+1freeeigenstates}
	\psi_k(x) = 
	\frac{\beta^{1/4}}{\sqrt{2\pi(k^2 + \beta)}}
	\begin{pmatrix}
		k e^{i \sqrt{\beta} kx} \\
		\sqrt{\beta} e^{i \sqrt{\beta} kx}
	\end{pmatrix}
\end{equation}
so that \( \psi_k(x) \) and \( \psi_{-k}(x) \) are the linearly independent eigenvectors corresponding to \( E = k^2 \).
(We can easily substitute the eigenfunction \( \psi_k(x) \) into \eqref{eq:1+1TimeIndependentLevyLeblond} to verify that \( \psi_k \) is indeed a solution).
Note that these eigenvectors are not square-integrable, so do not live inside the Hilbert space \( \Lp{2}(\RR)\otimes\CC^2 \).
Regardless, we find that
\begin{equation}\label{eq:1+1orthonormaleigvec}
	\begin{aligned}
		\int_{-\infty}^\infty \dd{x} \adjt{\psi_j}(x) \psi_k(x) 
		&= \frac{jk + \beta}{\sqrt{(k^2 + \beta)(j^2 + \beta)}} \frac{\sqrt{\beta}}{2\pi} \int_{-\infty}^\infty \dd{x} e^{i \sqrt{\beta} (k-j) x}\\
		&= \frac{jk + \beta}{\sqrt{(k^2 + \beta)(j^2 + \beta)}} \delta(k-j)\\
		&= \delta(k-j)
	\end{aligned}
\end{equation}
where \( \delta \) is the Dirac delta function.
Above, we used the fact that the Fourier transform of \( \delta \) is the constant function \( 1 \).

Equation~\eqref{eq:1+1orthonormaleigvec} tells us that the eigenstates \( \psi_k \) are orthonormal.
The solutions to the time-dependent L\'evy-Leblond equation~\eqref{eq:TimeDependentLevyLeblond} corresponding to these eigenvectors are
\[
	\Psi_k(t, x) = e^{-ik^2t}\psi_k(x)
	=
	\frac{\beta^{1/4}}{\sqrt{2\pi(k^2 + \beta)}}
	\begin{pmatrix}
		k e^{i \sqrt{\beta} kx - ik^2 t} \\
		\sqrt{\beta} e^{i \sqrt{\beta} kx - ik^2 t}
	\end{pmatrix}.
\]
We can find a more general solution to the time-dependent L\'evy-Leblond equation by taking a continuous linear combination:
\[
	\Psi(t,x) = \int_{-\infty}^\infty \dd{k} f(k) \Psi_k(t,x)
\]
for some function \( f\in \Lp{1}(\RR)\cap\Lp{2}(\RR) \).
Note that, although the eigenstate solutions \( \Psi_k(t,\cdot) \) are not in \( \Lp{2}(\RR)\otimes\CC^2 \) for any \( t\in\RR \),
we can choose \( f \) so that the linear combination \( \Psi(t,\cdot) \) will be in \( \Lp{2}(\RR)\otimes\CC^2 \) for each \( t\in\RR \).

We have just demonstrated how the operators of the colour Lie superalgebra \( \mathfrak{D} \) 
can be used to solve for the eigenstates~\eqref{eq:1+1freeeigenstates} of the free L\'evy-Leblond equation.
This demonstrates the fundamental nature of the \( \Ztwo[3] \)-colour Lie superalgebra \( \mathfrak{D} \) to this L\'evy-Leblond equation.

\section{Ladder operators of the harmonic oscillator}\label{sec:harmonic}
The \( \Ztzt \)-graded symmetry of the L\'evy-Leblond equation with harmonic potential has been investigated in~\cite{AKTT2016},
where the constructed symmetry algebra did not have a \( \Ztzt \)-grade (contrasting with the free equation).
Despite this, we construct a \( \Ztzt \)-graded spectrum generating algebra for the harmonic L\'evy-Leblond equation by using a different form of the L\'evy-Levlond operator:
from~\eqref{eq:TimeIndependentLevyLeblond} (and assuming the relations~\eqref{eq:gammaMatricesFreeEquation} hold), the time independent equation with potential \( V(x) = (k/2)x^2 \)
 is
\begin{equation}\label{eq:1+1TimeIndependentLevyLeblondHarmonic}
	\left(\gamma_- \beta - \gamma^1 i \partial_1 + \gamma_{+} \frac{k}{2}x^{2}\right) \ket{\psi} = \gamma_+ E \ket{\psi}
\end{equation}
where \( k \) is the spring constant.
Note that we are using natural units \( \hbar = 1 \).

The L\'evy-Leblond Hamiltonian is
\[
	\Hll = \gamma_{-} \beta - \gamma^{1} i\partial_1 + \gamma_{+} \frac{k}{2}x^{2}
\]
and the corresponding Schr\"odinger-like Hamiltonian is
\[
	\Hsch = \frac{1}{\beta}(\Hll)^2 = - \frac{1}{\beta}(\partial_1)^2 +  \frac{k}{2} x^{2} - \gamma_{+} i \frac{k}{\beta} x.
\]
From \Cref{prop:samespectrum}, the spectrum of \( \Hll \) is contained in the spectrum of the familiar harmonic Schr\"odinger Hamiltonian.
Therefore, the spectrum of \( \Hll \) is positive and discrete.

Consider a general first order differential operator
\[
	A = (c_I(x) + c_{I1}(x)\partial_1)I + (c_+(x)+c_{+1}(x)\partial_1)\gamma_+ + (c_-(x)+c_{-1}(x)\partial_1)\gamma_- + (c_1(x)+c_{11}(x)\partial_1)\gamma^1.
\]
Our strategy is to search for ladder operators, so we want
\begin{equation}\label{eq:LadderOperatorEquation}
	\Hll A \ket{\psi} = \gamma_+(E + F)A\ket{\psi}
\end{equation}
for \( E,F\in\CC \) and
for \( \ket{\psi} \) a solution to the time-independent equation~\eqref{eq:1+1TimeIndependentLevyLeblondHarmonic}.


If \( \comm{\Hll - \gamma_+E}{A} = \gamma_+ FA \) then \( A \) will satisfy \eqref{eq:LadderOperatorEquation}.
Indeed,
\[
	(\Hll - \gamma_+E)A\ket{\psi}
	= A(\Hll - \gamma_+E)\ket{\psi} + \comm{\Hll - \gamma_+E}{A}\ket{\psi}
	= \gamma_+ FA\ket{\psi}.
\]
Similarly, if \( \acomm{\Hll - \gamma_+E}{A} = \gamma_+ FA \) then \( A \) will also satisfy \eqref{eq:LadderOperatorEquation}.

\begin{rmk}
	Recall that, for the Schr\"odinger Hamiltonian \( H \), if \( \comm{H}{A} = FA \) then \( A \) is a Ladder operator.
	For the L\'evy-Leblond Hamiltonian \( \Hll \), we also need to consider the \( \gamma_+E \) term
	because we don't necessarily have \( \comm{\gamma_+E}{A} = 0 \).
\end{rmk}

In the case that \( \comm{\Hll - \gamma_+E}{A} = \gamma_+ FA \), we obtain the following ODE,
\begin{equation}
	\left\{
		\begin{aligned}
			\dv{c_{1}}{x} &= \frac{i \left(F c_{-}{\left(x \right)} + k x c_{-1}{\left(x \right)}\right)}{2}\\
			\dv{c_{11}}{x} &= \frac{i F c_{-1}{\left(x \right)}}{2}\\
			\dv{c_{I}}{x} &= \frac{i \left(2 E c_{-}{\left(x \right)} + F c_{-}{\left(x \right)} + 2 \beta c_{+}{\left(x \right)} - k x^{2} c_{-}{\left(x \right)} - k x c_{-1}{\left(x \right)}\right)}{2}\\
			\dv{c_{I1}}{x} &= i \left(E c_{-1}{\left(x \right)} + \frac{F c_{-1}{\left(x \right)}}{2} + \beta c_{+1}{\left(x \right)} - \frac{k x^{2} c_{-1}{\left(x \right)}}{2}\right)\\
			\dv{c_{+}}{x} &= i \left(- 2 E c_{1}{\left(x \right)} - F c_{1}{\left(x \right)} + F c_{I}{\left(x \right)} + k x^{2} c_{1}{\left(x \right)} + k x c_{11}{\left(x \right)} + k x c_{I1}{\left(x \right)}\right)\\
			\dv{c_{+1}}{x} &= - 2 i E c_{11}{\left(x \right)} - i F c_{11}{\left(x \right)} + i F c_{I1}{\left(x \right)} + i k x^{2} c_{11}{\left(x \right)} - 2 c_{+}{\left(x \right)}\\
			\dv{c_{-}}{x} &= 2 i \beta c_{1}{\left(x \right)}\\
			\dv{c_{-1}}{x} &= 2 i \beta c_{11}{\left(x \right)} - 2 c_{-}{\left(x \right)}\\
			c_{+1}{\left(x \right)} &= 0\\
			c_{-1}{\left(x \right)} &= 0
		\end{aligned}
	\right.
\end{equation}
while in the case that \( \acomm{\Hll - \gamma_+E}{A} = \gamma_+ FA \),
we obtain 
\begin{equation}
	\left\{
		\begin{aligned}
			\dv{c_{1}}{x} &= \frac{i \left(2 E c_{-}{\left(x \right)} + F c_{-}{\left(x \right)} - 2 \beta c_{+}{\left(x \right)} - k x^{2} c_{-}{\left(x \right)} - k x c_{-1}{\left(x \right)}\right)}{2}\\
			\dv{c_{11}}{x} &= i E c_{-1}{\left(x \right)} + \frac{i F c_{-1}{\left(x \right)}}{2} - i \beta c_{+1}{\left(x \right)} - \frac{i k x^{2} c_{-1}{\left(x \right)}}{2} - 2 c_{1}{\left(x \right)}\\
			\dv{c_{I}}{x} &= \frac{i \left(F c_{-}{\left(x \right)} + k x c_{-1}{\left(x \right)}\right)}{2}\\
			\dv{c_{I1}}{x} &= \frac{i F c_{-1}{\left(x \right)}}{2} - 2 c_{I}{\left(x \right)}\\
			\dv{c_{+}}{x} &= i \left(2 E c_{I}{\left(x \right)} - F c_{1}{\left(x \right)} + F c_{I}{\left(x \right)} - k x^{2} c_{I}{\left(x \right)} - k x c_{11}{\left(x \right)} - k x c_{I1}{\left(x \right)}\right)\\
			\dv{c_{+1}}{x} &= i \left(2 E c_{I1}{\left(x \right)} - F c_{11}{\left(x \right)} + F c_{I1}{\left(x \right)} - k x^{2} c_{I1}{\left(x \right)}\right)\\
			\dv{c_{-}}{x} &= 2 i \beta c_{I}{\left(x \right)}\\
			\dv{c_{-1}}{x} &= 2 i \beta c_{I1}{\left(x \right)}\\
			c_{11}{\left(x \right)} &= 0\\
			c_{I1}{\left(x \right)} &= 0
		\end{aligned}
	\right.
\end{equation}

Solving these two equations, we obtain the following operators and (anti)commutation relations,
\begin{align*}
	\Hll &=  \gamma^{1} (-i\partial_1) +  \gamma_{-} \beta + \gamma_{+} \frac{k}{2}x^{2} 
	     & \comm{\Hll}{\Hll} &= 0\\
	b &= \sqrt{\frac{\beta k}{2}}x + \partial_1 - \gamma_{+}i \sqrt{\frac{k}{2\beta}}
	  &\comm{\Hll}{b} &= -\sqrt{\frac{2k}{\beta}}\gamma_+ b\\
	b^{\dagger} &= \sqrt{\frac{\beta k}{2}}x - \partial_1 - \gamma_{+}i \sqrt{\frac{k}{2\beta}}
		    &\comm{\Hll}{b^{\dagger}} &= \sqrt{\frac{2k}{\beta}}\gamma_+ b^{\dagger}\\
	c &= \gamma_{+} 2i\partial_1 + \gamma^{1} \beta = \beta \Id - 2\gamma_+\Hll
	  & \acomm{\Hll}{c} &= 0
\end{align*}
From the above relations,
we obtain two additional operators \( \gamma_+b \) and \( \gamma_+b^{\dagger} \).
We can easily compute the anticommutation relations for \( \gamma_+b \) and \( \gamma_+b^{\dagger} \) by using the derivation property
\begin{align*}
	\acomm{\Hll}{\gamma_+ b} &= \acomm{\Hll}{\gamma_+}b - \gamma_+\comm{\Hll}{b} = \beta b\\
	\acomm{\Hll}{\gamma_+ b^{\dagger}} &= \acomm{\Hll}{\gamma_+}b^{\dagger} - \gamma_+\comm{\Hll}{b^{\dagger}} = \beta b^{\dagger}.
\end{align*}

Note that the operators \( b,\,b^{\dagger},\,\gamma_+b \) and \( \gamma_+b^{\dagger} \) are all ladder operators for \( \Hsch \):
\begin{align*}
	\comm{\Hsch}{b} &= -\sqrt{\frac{2k}{\beta}}b & \comm{\Hsch}{\gamma_+b^{\dagger}} &= \sqrt{\frac{2k}{\beta}} b^{\dagger}\\
	\comm{\Hsch}{\gamma_+b} &= -\sqrt{\frac{2k}{\beta}}\gamma_+ b & \comm{\Hsch}{\gamma_+b^{\dagger}} &= \sqrt{\frac{2k}{\beta}}\gamma_+ b^{\dagger}.
\end{align*}
.
However, \( \gamma_+b \) and \( \gamma_+b^{\dagger} \) are \emph{not} ladder operators for \( \Hll \).
This emphasises a key distinction between the L\'evy-Leblond equation and the Schr\"odinger equation:
some solutions of the Schr\"odinger equation will not be solutions of the L\'evy-Leblond equation.

The other (anti)commutation relations between these operators that are of interest to us are
\begin{align*}
	\acomm{\Hll}{\Hll} = \acomm{b}{b^{\dagger}} &= 2\beta\Hsch\\
	\acomm{\gamma_+}{\Hll} &= \beta\Id\\
	\comm{\gamma_+}{b} =  \comm{\gamma_+}{b^{\dagger}} = \acomm{\gamma_+}{c} &= 0\\
	\comm{b}{c} = \comm{b^{\dagger}}{c} = \acomm{c}{\gamma_+ b} = \acomm{c}{\gamma_+ b^{\dagger}} &= 0\\
	\acomm{c}{c} &= 2\beta^{2}\Id\\
	\acomm{b}{\gamma_+b} &= 2\gamma_+b^2\\
	\acomm{b^{\dagger}}{\gamma_+b^{\dagger}} &= 2\gamma_+(b^{\dagger})^2\\
	\acomm{b}{\gamma_+b^{\dagger}} = \acomm{b^{\dagger}}{\gamma_+b} &= 2\gamma_+\beta\Hsch\\
	\comm{\gamma_+b}{\gamma_+b^{\dagger}} &= 0
\end{align*}
as well as \( \comm{b}{b^{\dagger}} = \sqrt{2\beta k} \Id \).

If we want to include \( \acomm{b}{b^{\dagger}} = 2\beta\Hsch \),
then we are unable to realise the (anti)commutation relations using only a \( \Ztwo \)-graded Lie superalgebra,
since we also need \( \comm{\Hll}{b}\), \( \comm{\Hll}{b^{\dagger}} \) but \( \acomm{\Hll}{\gamma_+b} \), \( \acomm{\Hll}{\gamma_+b^{\dagger}} \).
These (anti)commutation relations instead require (at least) a \( \Ztzt \)-graded colour algebra.
To this end, define \( \mathfrak{L} = \mathfrak{L}_{00}\oplus\mathfrak{L}_{01}\oplus\mathfrak{L}_{10}\oplus\mathfrak{L}_{11} \) by
\begin{align*}
	\mathfrak{L}_{00} &= \spn\{\Id,\Hsch, b^2, (b^{\dagger})^2\} &
	\mathfrak{L}_{01} &= \spn\{b,b^{\dagger}\}\\
	\mathfrak{L}_{10} &= \spn\{\Hll, c, \gamma_+, \gamma_+\Hsch, \gamma_+b^2, \gamma_+(b^{\dagger})^2\} &
	\mathfrak{L}_{11} &= \spn\{\gamma_+b,\gamma_+b^{\dagger}\}.
\end{align*}
We can easily verify that \( \mathfrak{L} \) closes to form a \( \Ztzt \)-graded colour Lie superalgebra,
and realises the (anti)commutation relations listed above (besides, of course, \( \comm{b}{b^{\dagger}} = \sqrt{2\beta k}\Id \)).

However, if we instead include the relation for \( \comm{b}{b^{\dagger}} \) we can then realise the (anti)commutation relations as a (\( \Ztwo \)-graded) Lie superalgebra
with even sector spanned by \( \Hsch,b,b^{\dagger},\Id \) and odd sector spanned by \( \Hll, \gamma_+b, \gamma_+b^{\dagger} \).
To summarise: including the relation for \( \comm{b}{b^{\dagger}} \) requires a \( \Ztwo \)-graded Lie superalgebra, whereas including the relation \( \acomm{b}{b^{\dagger}} \) requires a \( \Ztzt \)-graded colour Lie superalgebra.
This is analogous to the situation for the harmonic Schr\"odinger equation:
we can realise the relation \( \comm{b}{b^{\dagger}} \) using an ungraded Lie algebra, but we require a \( \Ztwo \)-graded Lie superalgebra to realise the relation \( \acomm{b}{b^{\dagger}} \).
Both of these relations are useful for determining the \( \gamma_+ \)-eigenvalues and \( \gamma_+ \)-eigenstates, as in the following theorem.

\begin{thm}
	For the harmonic L\'evy-Leblond Hamiltonian \( \Hll \),
	the \( \gamma_+ \)-eigenvalues are \( \omega(n + 1/2)\)
	where \( \omega = \sqrt{2k/\beta} \) is the angular frequency
	and \( n \) is a non-negative integer.
	Moreover, the corresponding eigenstates are
	\[
		(b^{\dagger})^n\vac
	\]
	where \( \vac \) is a vacuum state satisfying \( b\vac = 0 \) and \( \Hll\vac = (1/2)\omega\vac \).
\end{thm}
\begin{proof}
From~\Cref{prop:samespectrum},
the \( \gamma_+ \)-spectrum of \( \Hll \) is a subset of the spectrum of the corresponding Schr\"odinger Hamiltonian.
We wish to show that every eigenvalue of \( \Hsch \) is a \( \gamma_+ \)-eigenvalue of \( \Hll \).

We know that the eigenvalues of \( \Hsch \) are positive and discrete.
Thus, there exists some eigenvector \( \vacsch \) for \( \Hsch \) with \( b\vacsch = 0 \) (that is, \( \vacsch \) is a vacuum state for \( \Hsch \)).
Then,
\begin{align*}
	\Hsch\vacsch &= \frac{1}{2\beta}\acomm{b}{b^{\dagger}}\vacsch \\
		     &= \frac{1}{2\beta}(bb^{\dagger} + b^{\dagger}b)\vacsch \\
		     &= \frac{1}{2\beta}(2 b^{\dagger}b + \comm{b}{b^{\dagger}})\vacsch \\
		     &= \frac{1}{2}\sqrt{\frac{2k}{\beta}}\vacsch\\
		     &= \frac{1}{2}\omega\vacsch
\end{align*}
That is, \( (1/2)\omega \)  will be an eigenvalue of \( \Hsch \). 

Now, consider the state \( \vac = (\beta \Id + c)\vacsch \).
We choose this state because the operator \( \beta \Id + c \) has useful properties.
For one, it is nearly idempotent:
\[
	(\beta \Id + c)^2 = \beta^2\Id + 2\beta c + c^2 = 2\beta(\beta \Id + c)
\]
(recalling that \( c^2 = \beta^2\Id \)).
Moreover, by examining the (anti)commutation relations, \( \beta\Id + c \) commutes with \( \Hsch \).
We have chosen to write this operator in terms of the elements of the colour Lie superalgebra,
but we can simplify this expression:
\[
	\beta\Id + c = 2\beta I - 2\gamma_+\Hll = 2\beta I + 2\Hll\gamma_+ - 2\acomm{\gamma_+}{\Hll} = 2\Hll\gamma_+.
\]
Therefore,
\begin{align*}
	\Hll\vac &= \Hll \frac{1}{2\beta}(\beta\Id + c)^2\vacsch\\
		 &= \frac{1}{\beta}\Hll\Hll\gamma_+(\beta\Id + c)\vacsch\\
		 &= \Hsch \gamma_+(\beta\Id + c)\vacsch\\
		 &=  \gamma_+(\beta\Id + c)\Hsch\vacsch\\
		 &=  \gamma_+(\beta\Id + c)\frac{1}{2}\omega\vacsch\\
		 &= \gamma_+\frac{1}{2}\omega\vac
\end{align*}
That is, provided \( \vac\neq 0 \), we have shown that \( \vac \) is a vacuum state for \( \Hll \) with \( \gamma_+ \)-eigenvalue \( (1/2)\omega \).

Now we show that we can choose \( \vacsch \) so that \( \vac\neq 0 \).
In the case that \( \gamma_+\vacsch = 0 \), instead choose \( \vacschp = \Hll\vacsch  \) to be the vacuum state for \( \Hsch \).
Then,
\[
	b\vacschp = \Hll b\vacsch + \comm{b}{\Hll}\vacsch = 0
\]
so \( \vacschp \) is indeed a vacuum state.
Moreover,
\[
	\gamma_+\vacschp = -\Hll\gamma_+\vacsch + \acomm{\gamma_+}{\Hll}\vacsch = \beta\vacsch \neq 0.
\]

That is, without loss of generality, we can assume that \( \gamma_+\vacsch \neq 0 \).
Since we are assuming that \( \gamma_+\vacsch \neq 0 \), we must also have \( \vac = 2\Hll\gamma_+\vacsch\neq 0 \),
for otherwise \( \gamma_+\vacsch \) would be a \( \gamma_+ \)-eigenvector for \( \Hll \) with \( \gamma_+ \)-eigenvalue \( 0 \),
but \( 0 \) is not in the \( \gamma_+ \)-spectrum for \( \Hll \).

We have thus shown that \( \ket{0} \) is an vacuum state with \( \gamma_+ \)-eigenvalue \( (1/2)\omega \).
Using the commutation relation \( \comm{\Hll}{b^{\dagger}} = \omega\gamma_+b^{\dagger} \), we obtain the eigenstates \( (b^{\dagger})^n\vac \) with \( \gamma_+ \)-eigenvalues \( \omega(n + 1/2) \).
Moreover, all eigenstates can be obtained in this way: taking an arbitrary eigenstate, we can apply the lowering operator \( b \) until a vacuum state is reached and then use \( b^{\dagger} \) to raise back up to this arbitrary eigenstate.
\end{proof}

Of particular note is the appearance of the operators from the colour Lie superalgebra \( \mathfrak{L} \) in the above proof.
This demonstrates the utility of this colour Lie superalgebra in analysing the L\'evy-Leblond equation.
However, unlike the colour Lie superalgebras of~\cite{AKTT2016, AKTT2017},
\( \mathfrak{L} \) contains operators that are not symmetry operators.
In summary, despite the super Schr\"odinger symmetry algebra for the harmonic potential not admitting a \( \Ztzt \)-graded extension~\cite{AKTT2016},
there does still exist a (non-symmetry) algebra \( \mathfrak{L} \) that is useful for computing the spectrum.

\section{Conclusions}
We have extended the results of~\cite{AKTT2016,AKTT2017} by showing that higher gradings (in particular, a \( \Ztwo[3] \)-grading) appear for the free L\'evy-Leblond equation
and that the harmonic potential also admits a \( \Ztzt \)-graded colour Lie superalgebra.

In~\cite{AKTT2016,AKTT2017}, the authors only searched for first- and second-order differential operators.
In this paper, we did not restrict ourselves to differential operators,
and searched for any operators which left the \( \gamma_+ \)-eigenspaces of the free L\'evy-Leblond equation invariant.
By carefully examining eigenvalue constraints,
we found operators that closed to form the \( \Ztwo[3] \)-graded colour Lie superalgebra \( \mathfrak{D} \).
Of particular note is the non-differential operator \( \gamma^1\parity\in\mathfrak{D} \) where \( \parity \) is the parity operator.
This inclusion of a non-differential operator forced the grading to be \( \Ztwo[3] \),
and it is likely that broader symmetry algebras will require \( \Ztwo[3] \) or even higher \( \Ztwo[n] \)-gradings (\( n>3 \)).
We then used \( \mathfrak{D} \) to solve the L\'evy-Leblond equation,
demonstrating the fundamental nature of the \( \Ztwo[3] \)-graded colour Lie superalgebra \( \mathfrak{D} \) to the L\'evy-Leblond equation.

The super Schr\"odinger symmetry algebra for the harmonic L\'evy-Leblond equation in \( (1+1) \)-dimensions admits no \( \Ztzt \)-graded extension~\cite{AKTT2016}.
Instead, we considered non-symmetry operators.
In particular,
we discovered a \( \Ztzt \)-graded colour Lie superalgebra
\( \mathfrak{L} \) that is generated by a Hamiltonian-like operator with the ladder operators.
Despite \( \mathfrak{L} \) containing non-symmetry operators,
the algebra has utility:
we used the operators from \( \mathfrak{L} \) to compute the spectrum of the harmonic L\'evy-Leblond equation.
To the author's knowledge, \( \mathfrak{L} \) is the first colour Lie superalgebra that has been found for a L\'evy-Leblond equation with non-vanishing potential.

These results further highlight the utility of \( \Ztzt \)-graded colour Lie superalgebras in examining the L\'evy-Leblond equation, and potentially other physical systems.

%
\subsection*{Acknowledgements}
This work was partially supported by the Australian Research Council through Discovery Project DP200101339.
The author is especially grateful to Phillip Isaac for proofreading the manuscript and for his continuing support and guidance.
The author would like to thank the anonymous reviewers for their helpful comments.

\appendix
\section{Non-relativistic limit of the Dirac equation}\label{sec:NonrelativisticLimit}
In this appendix, we show how the L\'evy-Leblond equation (corresponding to the Schr\"odinger equation) can be obtained as the non-relativistic limit of the Dirac equation.
This generalises the original argument of L\'evy-Leblond~\cite{LevyLeblond1967} and provides a realisation of the L\'evy-Leblond gamma matrices in \( \Cl{1}{d}(\RR)\otimes\CC \) when \( d \) is odd.

Recall that the Dirac equation is
\[
	(\widetilde{\gamma}^0 (1/c)i\hbar\partial_t +  \widetilde{\gamma}^j i\hbar\partial_j - mc)\Psi(t,\vec{x}) = 0
\]
and is a relativistic equation describing a free spin-\( 1/2 \) particle (see e.g.~\cite[Chapter~10]{Lounesto2001}). 
The gamma matrices appearing in the Dirac equation satisfy \( \acomm{\widetilde{\gamma}^j}{\widetilde{\gamma}^k}= 2\eta_{jk},\, (j,k=0,\ldots,d) \),
where \( \eta_{jk} \) is a symmetric bilinear form with signature \( (1,d) \):
\[
	\eta_{jk} =
	\begin{cases}
		\delta_{jk} & \text{if}~j=0,\\
		-\delta_{jk} & \text{otherwise}.
	\end{cases}
\]
The gamma matrices for the Dirac equation are members of the complexified Clifford algebra \( \Cl{1}{d}(\RR)\otimes\CC \) .
From now on, we will work in natural units, \( c = \hbar = 1 \).

We can also write a time-independent version of the Dirac equation,
by replacing \( i\partial_t \) with the total energy \( \mathscr{E} \) (\( mc^2 + \) kinetic energy):
\begin{equation}\label{eq:timeindependentdirac}
	(\widetilde{\gamma}^0\mathscr{E} + \widetilde{\gamma}^j i\partial_j - m)\Psi(\vec{x}) = 0.
\end{equation}

Let us recall the original limit argument of L\'evy-Leblond~\cite{LevyLeblond1967}.
For the moment, we shall restrict ourselves to \( (1+3) \)-dimensions and choose
the Dirac representation of the gamma matrices in \( 4 \times 4 \) matrices:
\[
	\widetilde{\gamma}^0 = 
	\begin{pmatrix}
		\Id & 0 \\
		0 & -\Id
	\end{pmatrix},
	\qquad
	\widetilde{\gamma}^j =
	\begin{pmatrix}
		0 & \sigma^j\\
		-\sigma^j & 0
	\end{pmatrix}
\]
where \( \sigma^j \) are the Pauli matrices and \( \Id \) is the identity.
In this representation, the Dirac equation becomes a coupled differential equation
\begin{equation}\label{eq:coupleddirac}
	\left\{
	\begin{aligned}
		(\mathscr{E} - m) \varphi + \sigma^ji\partial_j \chi &= 0 \\
		(-\mathscr{E} - m) \chi - \sigma^ji\partial_j \varphi &= 0 \\
	\end{aligned}
	\right.
\end{equation}
where
\( \Psi = 
\begin{pmatrix}
	\varphi\\
	\chi
\end{pmatrix}
\) is a solution to the Dirac equation.
L\'evy-Leblond discovered \cite{LevyLeblond1967} that taking \( \mathscr{E} = E + m \) (where \( E \) is the non-relativistic kinetic energy), and assuming \( E \ll m \), we can approximate~\eqref{eq:coupleddirac} by
\begin{equation}\label{eq:diraclimit}
	\left\{
	\begin{aligned}
		E \varphi + \sigma^ji\partial_j \chi &= 0 \\
		-2m \chi - \sigma^ji\partial_j \varphi &= 0. \\
	\end{aligned}
	\right.
\end{equation}
If \( \varphi \) and \( \chi \) solve~\eqref{eq:diraclimit},
then \( \Psi = 
\begin{pmatrix}
	\chi \\
	\varphi
\end{pmatrix}
\) is a solution to the L\'evy-Leblond equation, with
\[
	\gamma_+ = 
	\begin{pmatrix}
		0 & \Id\\
		0  & 0
	\end{pmatrix},
	\qquad
	\gamma_- = 
	\begin{pmatrix}
		0 & 0\\
		\Id & 0
	\end{pmatrix},
	\qquad
	\gamma^j = 
	\begin{pmatrix}
		\sigma^j & 0\\
		0 & -\sigma^j
	\end{pmatrix}
\]
(notice that the components \( \varphi \) and \( \chi \) have swapped places when compared to the solution of the Dirac equation).
That is, the L\'evy-Leblond equation is the non-relativistic limit of the Dirac equation.

We can generalise the above argument of L\'evy-Leblond for all odd space dimensions in a representation independent way.
Consider the time-independent Dirac equation in \( (1+d) \)-dimensions, \eqref{eq:timeindependentdirac} and assume \( d \) is odd.
Let 
\[
	\chiralg = i^{(d+3)d/2}\prod_{j=0}^d \widetilde{\gamma}^j.
\]
Since the \( \widetilde{\gamma}^j \) anticommute,
we have that \( \widetilde{\gamma}^j\chiralg = (-1)^d \chiralg \widetilde{\gamma}^j,\, (j=0,\ldots,d) \).
Assuming that \( d \) is odd, we have that \( \acomm{\widetilde{\gamma}^j}{\chiralg} = 0 \).
Additionally,
\[
	\chiralg = i^{(d+3)d/2}(-1)^{d+ (d-1) + \cdots + 1} \prod_{j=0}^d \widetilde{\gamma}^{d-j} = i^{(d+3)d/2}(-1)^{(d+1)d/2} \prod_{j=0}^d \widetilde{\gamma}^{d-j}.
\]
Therefore, \( (\chiralg)^2 = (i^2)^{(d+3)d/2}(-1)^{(d+1)d/2 + d} \Id = \Id \).
In particular, \( \chiralg \) is invertible so \( \Psi \) is a solution to the time-independent Dirac equation if and only if
\begin{equation}\label{eq:chiralDirac}
	\chiralg(\widetilde{\gamma}^0\mathscr{E} + \widetilde{\gamma}^j i \partial_j - m)\Psi(\vec{x}) = 0.
\end{equation}
(The above multiplication by \( \chiralg \) corresponds to the swapping of components \( \varphi \) and \( \chi \) in L\'evy-Leblond's original argument above.)
Take the non-relativistic limit by setting \( \mathscr{E} = E + m \) and assuming \( E \ll m \).
We can compute
\begin{align}
	0 &= \chiralg(\widetilde{\gamma}^0\mathscr{E} + \widetilde{\gamma}^j i \partial_j - m)\Psi(\vec{x}) \nonumber\\
	  &= (\chiralg\widetilde{\gamma}^0(E+m) + \chiralg\widetilde{\gamma}^j i \partial_j - \chiralg m)\Psi(\vec{x}) \nonumber\\
	  &= \left(\frac{1}{2}(\chiralg + \chiralg\widetilde{\gamma}^0)E - \frac{1}{2}(\chiralg - \chiralg\widetilde{\gamma}^0)(E+2m) + \chiralg\widetilde{\gamma}^j i \partial_j\right)\Psi(\vec{x}) \nonumber\\
	  &\approx (\gamma_+ E - \gamma_- 2m + \gamma^j i \partial_j)\Psi(\vec{x}) \label{eq:derivedTimeIndependentLevyLeblond}
\end{align}
where
\begin{equation}\label{eq:LevyLeblondGammaFromDirac}
	\gamma_{\pm} = \frac{1}{2}(\chiralg \pm \chiralg\widetilde{\gamma}^0),
	\qquad
	\gamma^j = \chiralg \widetilde{\gamma}^j,\, (j=1,\ldots,d).
\end{equation}
It is easily verified that \( \gamma_{\pm},\, \gamma^j \) satisfy the anti-commutation relations~\eqref{eq:LevyLeblondAR}
by using the fact that \( \acomm{\widetilde{\gamma}^j}{\chiralg} = 0 \) (when \( d \) is odd).
We can now use the time independent L\'evy-Leblond equation~\eqref{eq:derivedTimeIndependentLevyLeblond} to derive a time-dependent L\'evy-Leblond equation by replacing \( E \) with \( i\partial_t \):
\[
	(\gamma_+\partial_t  + \gamma_- i2m + \gamma^j \partial_j)\Psi(t,\vec{x}) = 0
\]
which has the desired form~\eqref{eq:generalLevyLeblond}.

It is important to note that the realisation in~\eqref{eq:LevyLeblondGammaFromDirac} satisfies relations (such as \( \gamma_+ \gamma^j + \gamma_- \gamma^j = \gamma^j \)) that are not satisfied by every representation of~\eqref{eq:LevyLeblondAR}.
In particular, when \( d=1 \), relations~\eqref{eq:gammaMatricesFreeEquation} are satisfied.
Despite being less general, the realisation in~\eqref{eq:LevyLeblondGammaFromDirac} has the advantage that the gamma matrices can be defined in terms of a bilinear form with signature \( (1,d) \), representing \( 1 \) time dimension and \( d \) spacial dimensions.


\begin{thebibliography}{10}
\expandafter\ifx\csname url\endcsname\relax
  \def\url#1{\texttt{#1}}\fi
\expandafter\ifx\csname urlprefix\endcsname\relax\def\urlprefix{URL }\fi
\expandafter\ifx\csname href\endcsname\relax
  \def\href#1#2{#2} \def\path#1{#1}\fi

\bibitem{LevyLeblond1967}
J.-M. L\'{e}vy-Leblond, Nonrelativistic particles and wave equations, Comm.
  Math. Phys. 6~(4) (1967) 286--311.
\newblock \href {https://doi.org/10.1007/BF01646020}
  {\path{doi:10.1007/BF01646020}}.

\bibitem{Dirac1928}
P.~A.~M. Dirac, The quantum theory of the electron, Proc. R. Soc. A 117~(778)
  (1928) 610--624.
\newblock \href {https://doi.org/10.1098/rspa.1928.0023}
  {\path{doi:10.1098/rspa.1928.0023}}.

\bibitem{AKTT2016}
N.~Aizawa, Z.~Kuznetsova, H.~Tanaka, F.~Toppan,
  {$\mathbb{Z}_2\times\mathbb{Z}_2$}-graded {L}ie symmetries of the
  {L}\'{e}vy-{L}eblond equations, PTEP. Prog. Theor. Exp. Phys.~(12) (2016)
  123A01, 26.
\newblock \href {https://doi.org/10.1093/ptep/ptw176}
  {\path{doi:10.1093/ptep/ptw176}}.

\bibitem{AKTT2017}
N.~Aizawa, Z.~Kuznetsova, H.~Tanaka, F.~Toppan, Generalized supersymmetry and
  the l{\'e}vy-leblond equation, in: S.~Duarte, J.-P. Gazeau, S.~Faci,
  T.~Micklitz, R.~Scherer, F.~Toppan (Eds.), Physical and Mathematical Aspects
  of Symmetries, Springer International Publishing, Cham, 2017, pp. 79--84.

\bibitem{Toppan2015}
F.~Toppan, \href{https://dx.doi.org/10.1088/1742-6596/597/1/012071}{Symmetries
  of the {S}chrödinger equation and algebra/superalgebra duality}, J. Phys.
  Conf. Ser. 597~(1) (2015) 012071.
\newblock \href {https://doi.org/10.1088/1742-6596/597/1/012071}
  {\path{doi:10.1088/1742-6596/597/1/012071}}.
\newline\urlprefix\url{https://dx.doi.org/10.1088/1742-6596/597/1/012071}

\bibitem{RW1978a}
V.~Rittenberg, D.~Wyler, Generalized superalgebras, Nuclear Phys. B 139~(3)
  (1978) 189--202.
\newblock \href {https://doi.org/10.1016/0550-3213(78)90186-4}
  {\path{doi:10.1016/0550-3213(78)90186-4}}.

\bibitem{RW1978b}
V.~Rittenberg, D.~Wyler, Sequences of {$Z_{2}\oplus Z_{2}$} graded {L}ie
  algebras and superalgebras, J. Math. Phys. 19~(10) (1978) 2193--2200.
\newblock \href {https://doi.org/10.1063/1.523552}
  {\path{doi:10.1063/1.523552}}.

\bibitem{Ree1960}
R.~Ree, Generalized {L}ie elements, Canadian J. Math. 12 (1960) 493--502.
\newblock \href {https://doi.org/10.4153/CJM-1960-044-x}
  {\path{doi:10.4153/CJM-1960-044-x}}.

\bibitem{BD2020a}
A.~J. Bruce, S.~Duplij, Double-graded supersymmetric quantum mechanics, J.
  Math. Phys. 61~(6) (2020) 063503, 13.
\newblock \href {https://doi.org/10.1063/1.5118302}
  {\path{doi:10.1063/1.5118302}}.

\bibitem{AKT2020}
N.~Aizawa, Z.~Kuznetsova, F.~Toppan, $\mathbb{Z}_2\times \mathbb{Z}_2$-graded
  mechanics: the classical theory, Eur. Phys. J. C 80 (2020).
\newblock \href {https://doi.org/10.1140/epjc/s10052-020-8242-x}
  {\path{doi:10.1140/epjc/s10052-020-8242-x}}.

\bibitem{AKT2021}
N.~Aizawa, Z.~Kuznetsova, F.~Toppan, $\mathbb{Z}_2\times \mathbb{Z}_2$-graded
  mechanics: The quantization, Nuclear Phys. B 967 (2021) 115426.
\newblock \href {https://doi.org/10.1016/j.nuclphysb.2021.115426}
  {\path{doi:10.1016/j.nuclphysb.2021.115426}}.

\bibitem{AAD2020a}
N.~Aizawa, K.~Amakawa, S.~Doi, {$\mathcal{N}$}-extension of double-graded
  supersymmetric and superconformal quantum mechanics, J. Phys. A 53~(6) (2020)
  065205, 14.
\newblock \href {https://doi.org/10.1088/1751-8121/ab661c}
  {\path{doi:10.1088/1751-8121/ab661c}}.

\bibitem{AAD2020b}
N.~Aizawa, K.~Amakawa, S.~Doi, {$\mathbb{Z}_2^n$}-graded extensions of
  supersymmetric quantum mechanics via {C}lifford algebras, J. Math. Phys.
  61~(5) (2020) 052105, 13.
\newblock \href {https://doi.org/10.1063/1.5144325}
  {\path{doi:10.1063/1.5144325}}.

\bibitem{DA2021}
S.~Doi, N.~Aizawa, {$\mathbb{Z}_2^3$}-graded extensions of {L}ie superalgebras
  and superconformal quantum mechanics, SIGMA Symmetry Integrability Geom.
  Methods Appl. 17 (2021) Paper No. 071, 14.
\newblock \href {https://doi.org/10.3842/SIGMA.2021.071}
  {\path{doi:10.3842/SIGMA.2021.071}}.

\bibitem{AAD2021}
N.~Aizawa, K.~Amakawa, S.~Doi, Color algebraic extension of supersymmetric
  quantum mechanics, in: Quantum theory and symmetries, CRM Ser. Math. Phys.,
  Springer, Cham, 2021, pp. 199--207.
\newblock \href {https://doi.org/10.1007/978-3-030-55777-5\_18}
  {\path{doi:10.1007/978-3-030-55777-5\_18}}.

\bibitem{Bruce2021}
A.~J. Bruce, Is the {$\mathbb{Z}_2\times\mathbb{Z}_2$}-graded sine-{G}ordon
  equation integrable?, Nuclear Phys. B 971 (2021) Paper No. 115514, 10.
\newblock \href {https://doi.org/10.1016/j.nuclphysb.2021.115514}
  {\path{doi:10.1016/j.nuclphysb.2021.115514}}.

\bibitem{Quesne2021}
C.~Quesne, Minimal bosonization of double-graded quantum mechanics, Modern
  Phys. Lett. A 36~(33) (2021) Paper No. 2150238, 8.
\newblock \href {https://doi.org/10.1142/S0217732321502382}
  {\path{doi:10.1142/S0217732321502382}}.

\bibitem{AIKT2023}
N.~Aizawa, R.~Ito, Z.~Kuznetsova, F.~Toppan,
  \href{https://doi.org/10.1016/j.nuclphysb.2023.116202}{New aspects of the
  {$\mathbb{Z}_2 \times \mathbb{Z}_2$}-graded {$1D$} superspace: induced
  strings and {$2D$} relativistic models}, Nuclear Phys. B 991 (2023) Paper No.
  116202, 27.
\newblock \href {https://doi.org/10.1016/j.nuclphysb.2023.116202}
  {\path{doi:10.1016/j.nuclphysb.2023.116202}}.
\newline\urlprefix\url{https://doi.org/10.1016/j.nuclphysb.2023.116202}

\bibitem{YJ2001}
W.~Yang, S.~Jing, A new kind of graded {L}ie algebra and parastatistical
  supersymmetry, Sci. China Ser. A 44~(9) (2001) 1167--1173.
\newblock \href {https://doi.org/10.1007/BF02877435}
  {\path{doi:10.1007/BF02877435}}.

\bibitem{JYL2001}
S.-C. Jing, W.-M. Yang, P.~Li, Graded {L}ie algebra generating of
  parastatistical algebraic relations, Commun. Theor. Phys. (Beijing) 36~(6)
  (2001) 647--650.
\newblock \href {https://doi.org/10.1088/0253-6102/36/6/647}
  {\path{doi:10.1088/0253-6102/36/6/647}}.

\bibitem{KHA2011a}
K.~Kanakoglou, A.~Herrera-Aguilar, Ladder operators, {F}ock-spaces,
  irreducibility and group gradings for the relative parabose set algebra, Int.
  J. Algebra 5~(9-12) (2011) 413--427.

\bibitem{KHA2011b}
K.~Kanakoglou, A.~Herrera-Aguilar, Graded fock-like representations for a
  system of algebraically interacting paraparticles, J. Phys.: Conf. Ser.
  287~(1) (2011) 012037.
\newblock \href {https://doi.org/10.1088/1742-6596/287/1/012037}
  {\path{doi:10.1088/1742-6596/287/1/012037}}.

\bibitem{Tolstoy2014b}
V.~N. Tolstoy, Once more on parastatistics, Phys. Part. Nuclei Lett. 11 (2014)
  933--937.
\newblock \href {https://doi.org/10.1134/S1547477114070449}
  {\path{doi:10.1134/S1547477114070449}}.

\bibitem{SVdJ2018}
N.~I. Stoilova, J.~Van~der Jeugt, The {$\mathbb{Z}_2\times\mathbb{Z}_2$}-graded
  {L}ie superalgebra {$pso(2m+1|2n)$} and new parastatistics representations,
  J. Phys. A 51~(13) (2018) 135201, 17.
\newblock \href {https://doi.org/10.1088/1751-8121/aaae9a}
  {\path{doi:10.1088/1751-8121/aaae9a}}.

\bibitem{Toppan2021a}
F.~Toppan, {$\mathbb{Z}_2 \times \mathbb{Z}_2$}-graded parastatistics in
  multiparticle quantum {H}amiltonians, J. Phys. A 54~(11) (2021) Paper No.
  115203, 35.
\newblock \href {https://doi.org/10.1088/1751-8121/abe2f2}
  {\path{doi:10.1088/1751-8121/abe2f2}}.

\bibitem{Toppan2021b}
F.~Toppan, Inequivalent quantizations from gradings and
  {$\mathbb{Z}_2\times\mathbb{Z}_2$} parabosons, J. Phys. A 54~(35) (2021)
  Paper No. 355202, 21.
\newblock \href {https://doi.org/10.1088/1751-8121/ac17a5}
  {\path{doi:10.1088/1751-8121/ac17a5}}.

\bibitem{Zhang2023}
R.~Zhang, Para-spaces, their differential analysis and an application to
  green's quantisation (2023).
\newblock \href {http://arxiv.org/abs/2312.04250} {\path{arXiv:2312.04250}}.

\bibitem{SVdJ2024}
N.~I. Stoilova, J.~V.~d. Jeugt, Orthosymplectic
  {$\mathbb{Z}_2\times\mathbb{Z}_2$}-graded {L}ie superalgebras and
  parastatistics, J. Phys. A 57~(9) (2024) Paper No. 095202, 13.

\bibitem{BP2009}
Y.~Bahturin, D.~Pagon, Classifying simple color {L}ie superalgebras, in:
  Algebras, representations and applications, Vol. 483 of Contemp. Math., Amer.
  Math. Soc., Providence, RI, 2009, pp. 37--54.
\newblock \href {https://doi.org/10.1090/conm/483/09433}
  {\path{doi:10.1090/conm/483/09433}}.

\bibitem{Ryan2024}
M.~Ryan, Refining the grading of irreducible lie colour algebra representations
  (2024).
\newblock \href {http://arxiv.org/abs/2403.02855} {\path{arXiv:2403.02855}}.

\bibitem{BCDM2020}
S.~Bao, D.~Constales, H.~De~Bie, T.~Mertens, Solutions for the
  {L}\'{e}vy-{L}eblond or parabolic {D}irac equation and its generalizations,
  J. Math. Phys. 61~(1) (2020) 011509, 12.
\newblock \href {https://doi.org/10.1063/1.5135503}
  {\path{doi:10.1063/1.5135503}}.

\bibitem{Faustino2023}
N.~Faustino, On fractional semidiscrete {D}irac operators of {L}\'evy-{L}eblond
  type, Math. Nachr. 296~(7) (2023) 2758--2779.

\bibitem{Scheunert1979}
M.~Scheunert, Generalized {L}ie algebras, J. Math. Phys. 20~(4) (1979)
  712--720.
\newblock \href {https://doi.org/10.1063/1.524113}
  {\path{doi:10.1063/1.524113}}.

\bibitem{Lounesto2001}
P.~Lounesto, Clifford algebras and spinors, 2nd Edition, Vol. 286 of London
  Mathematical Society Lecture Note Series, Cambridge University Press,
  Cambridge, 2001.
\newblock \href {https://doi.org/10.1017/CBO9780511526022}
  {\path{doi:10.1017/CBO9780511526022}}.

\end{thebibliography}

\end{document}